\documentclass[11pt]{article}
\usepackage{verbatim} %comment

\usepackage{makeidx}
\usepackage{euscript}
\usepackage{amsmath,amssymb,amsfonts}
\usepackage{dsfont}
\usepackage{latexsym}
\usepackage{subfigure}
\usepackage{graphicx}
\usepackage{times}\usepackage[scaled=0.92]{helvet} % this is the recommended way of using \sf in times fonts
\usepackage{euler}
\usepackage{fancybox}
\usepackage{fullpage}
\usepackage{hyperref}
\usepackage{paralist}

\usepackage{color}
\usepackage{wrapfig}
\usepackage{tikz}

\usepackage[small,compact]{titlesec}

\usepackage[noend]{algorithmic}
\usepackage[linesnumbered,ruled,vlined]{algorithm2e}

\newcount\shortyear\newcount\shorthour\newcount\shortminute
\shorthour=\time\divide\shorthour by 60\shortyear=\shorthour
\multiply\shortyear by 60\shortminute=\time\advance\shortminute by
-\shortyear \shortyear=\year\advance\shortyear by -1900

\def\zeit{\number\shorthour:\ifnum\shortminute<10 0\number\shortminute
	\else\number\shortminute\fi}

\newenvironment{proof}{\noindent{\bf Proof : \ }}{\hfill$\Box$\par\medskip}

\newtheorem{theorem}{Theorem}
\newtheorem{corollary}[theorem]{Corollary}
\newtheorem{lemma}[theorem]{Lemma}

\newtheorem{definition}[theorem]{Definition}

\newenvironment{proofof}[1]{\begin{trivlist} \item {\bf Proof
			#1:~~}}
	{\qed\end{trivlist}}
\renewenvironment{proofof}[1]{\par\medskip\noindent{\bf Proof of #1: \ }}{\hfill$\Box$\par\medskip}

\newcommand{\COMMENTED}[1]{{}}

\newcommand{\eps}{\ensuremath{\varepsilon}}

\newcommand{\alg}{{\sf Alg}}

\newcommand{\prob}[1]{\operatorname{Pr}\left[#1\right]}
\newcommand{\ex}[1]{\operatorname{E}\left[#1\right]}

\newcommand{\X}{\ensuremath{\mathcal{X}}}

\newcommand{\our}{best-choice~}
\newcommand{\Tau}{\mathcal{T}}
\renewcommand{\P}{\mathcal{P}}
\newcommand{\maxall}{\max_{i=1}^{n}(x_i)}

\renewcommand{\i}{\ensuremath{\zeta}}
\renewcommand{\j}{\ensuremath{\eta}}

\newcommand{\eat}[1]{}

\DeclareMathOperator*{\argmax}{arg\,max}

\title{%Prophet Inequality and Prophet Secretary:
	Prophets, Secretaries, and Maximizing the Probability of \\ Choosing
	the Best}

\author{
	Hossein Esfandiari\thanks{Supported in part by NSF grants CCF-1320231 and CNS-1228598.}\\Google Research \\ New York,
	NY
	\and MohammadTaghi HajiAghayi   \thanks{Supported in part by NSF CAREER award CCF-1053605,  NSF BIGDATA grant IIS-1546108, NSF AF:Medium grant CCF-1161365,
		DARPA GRAPHS/AFOSR grant FA9550-12-1-0423, and another DARPA SIMPLEX grant. Part of this work was done while visiting Microsoft Research New England.}\\University of Maryland\\ College Park, MD
	\and
	Brendan Lucier \\Microsoft Research\\ Cambridge, MA
	\and
	Michael Mitzenmacher\thanks{Supported in part by  NSF grants CCF-1563710, CCF-1535795, CCF-1320231, and CNS-1228598.  Part of this work was done while visiting Microsoft Research New England.}\\Harvard University \\ Cambridge,
	MA
}

\date{}

\begin{document}
\sloppy
\maketitle
\thispagestyle{empty}

\begin{abstract}
Suppose a customer is faced with a sequence of fluctuating prices, such as for airfare or
a product sold by a large online retailer.  Given distributional information about what price
they might face each day, how should they choose when to purchase in order to maximize
the likelihood of getting the best price in retrospect?
This is related to the classical secretary problem, but with values drawn from known
distributions.
In their pioneering work, Gilbert and Mosteller [\textit{J. Amer. Statist. Assoc. 1966}]
showed that when the values are drawn i.i.d., there is a thresholding algorithm that
selects the best value with probability approximately $0.5801$.
However, the more general problem with non-identical distributions has remained unsolved.

In this
paper we 
provide an algorithm for the case of non-identical distributions that
selects the maximum element with probability $1/e$, and we show that
this is tight. We further show that if the observations arrive in a random
order, 
this barrier of $1/e$
can be broken using a static threshold algorithm, and we show that our 
success probability is the best possible for any single-threshold algorithm under random observation order.  Moreover, 
we prove that one can achieve a strictly better success probability using
more general multi-threshold algorithms, unlike the non-random-order case.
Along the way, we show that the best achievable success probability for
the random-order case matches that of the i.i.d.\ case, which is approximately $0.5801$,
under a 
``no-superstars'' condition that no single distribution is very likely
ex ante to generate the maximum value.
We also extend our results to the
problem of selecting one of the $k$ best values.

One of the main tools in our analysis is a suitable
``Poissonization'' of random order distributions, which uses
Le Cam's theorem to connect the Poisson binomial distribution with
the discrete Poisson distribution. This approach may be of independent interest.

\end{abstract}

% Gilbert and Mosteller~\cite{GM66} mention several
% applications for this fundamental problem  and in particular in the
% simplified atomic bomb inspection program where the player hunting
% with maximum probability is the detective nation and the placer of
% the maximum is the cheating nation.
% , we use Le Cam's
% Theorem as well as our own parallel version of it to handle discrete
% Poissons when we represent the joint probability distribution (of
% different distributions) for the maximum via lots of discrete
% Poissons, which is of its own interest.

\maketitle

\setcounter{page}{0}
\newpage

\section{Introduction}

%\textbf{[BL: I tried to further emphasize the algorithmic aspects of our results.  Can/should we emphasize this even further?]}
Suppose we are given a sequence of real numbers one by one, drawn from
independent but not necessarily identical distributions known in advance.
We can keep a single number from the sequence, but this choice must be made online.
At each observation, we can either select the current number
or push our luck and continue to the next observation.  Our goal is to 
maximize the probability of selecting the maximum (or equivalently minimum)
number from the sequence.

As a toy application, consider an airfare platform that provides a service
of suggesting when a buyer should purchase their ticket for the lowest fare.
Such a platform has distributional information about how expensive the fare 
will be each day before the flight.  Users hope to avoid the regret of purchasing
at a suboptimal price, and this incentivizes the platform to maximize the 
likelihood of suggesting the best price in hindsight.  
Given a model that maps time-before-flight and other fixed information (such
as location and airline) to a distribution over prices, how should the platform
make its online recommendations, and how likely is it to achieve the best price?

%As another example, consider a consumer making a purchase from an online retailer for an
%upcoming event, such as a child's birthday present.  The actual cost is relatively small 
%and inconsequential, but nevertheless the consumer is aware that the price fluctuates
%day by day and would be annoyed if the price drops after they purchase but before
%the party.\footnote{This example is a work of fiction.  Any resemblence to an actual scenario
%faced by one of the authors is purely coincidental.}  
%Such price fluctuations are perhaps less predictable than airfare as a function of time; 
%while the consumer has some beliefs about what sorts of prices they might face, 
%it is as though they are arriving in random order.  
%%But 
%%nevertheless the consumer might have a prior over, e.g., how frequently the price
%%will dip significantly below baseline.  
%Again, how should the consumer choose optimally
%to maximize the likelihood of getting the best price, and how likely are they to succeed?

\paragraph{Secretaries: A Related Problem}
This question is related to 
%These questions are related to 
%In 
the classical {\em secretary problem}.
%, perhaps the most well-known
%online stochastic optimization problem.
%in optimal stopping theory, 
In the secretary problem, we receive a sequence of randomly
permuted numbers 1 to $n$ in an online fashion.  We are given the
numbers one by one, but each time we observe a number, we see only
its relative rank compared to the previously observed numbers.  At
each observation, we have the option to stop the process and select
the most recent number. The goal is to maximize the probability of
selecting the maximum (or equivalently minimum) number.  For this problem,
Dynkin~\cite{dynkin1963optimum} presents a simple but elegant
algorithm that succeeds with probability at least $1/e$; indeed,
the success probability converges from above to $1/e$ as $n$ grows large,
and $1/e$ is the best possible bound (up to lower order terms) we can
achieve for this problem.

A natural variation of the problem assumes that the numbers are
drawn from the same known distribution, and the numbers themselves
are revealed one by one.  In their classic work, Gilbert and
Mosteller~\cite{GM66} consider this so-called ``full information'' case.\footnote{Despite the name, 
this is similar to the ``incomplete information'' setting in mechanism design.  
We will use the term ``full information'' in the sense of secretary 
problems, meaning that the distributions are known in advance.}
  As a starting
point, they show that one can pick a single threshold $\tau$ such that 
stopping at the first value larger than $\tau$ will select the maximum 
value with probability
approximately $0.517$ (asymptotically as $n$ grows large).  For the
general case where one can use a distinct threshold at each step,
they show that with the appropriate choice of thresholds one
succeeds at stopping at the maximum value with probability
approximately $0.5801$ (again, asymptotically as $n$ grows large;
both bounds are tight). These bounds significantly improve upon the
$1/e \approx 0.37$ result for the secretary problem, which
corresponds to the setting where the underlying distribution is not
known and only the relative ranks are
obtained.\footnote{Interestingly, one of the original motivations
Gilbert and Mosteller~\cite{GM66} provide is a simplified model for
an atomic bomb inspection program, which may have arisen from
Mosteller's work in Samuel Wilks's Statistical Research Group in New
York city during World War II on statistical questions about
airborne bombing.}

Since the work of Gilbert and Mosteller, there has been a vast literature on secretary problems going well 
beyond the scope of this paper.  We refer the interested reader to 
a survey by Freeman for an overview of this activity from the perspective of
stopping theory~\cite{Free83}.  The full information
case has received less attention, but there has been a notable line of work considering
variations such as $n$ being randomized~\cite{Porosinski87} and/or it being possible to revisit 
previously-observed values with a probability of failure~\cite{Petruccelli81}.  
To our knowledge, this literature on the full information
case has focused exclusively on the case of i.i.d.\ values.

\paragraph{Our Contributions}
We consider the more general problem of selecting the 
maximum (or minimum) value when the numbers are drawn from distributions that are 
independent but \emph{not necessarily identical}.
Our first result is that there is an algorithm that
achieves a success probability of $1/e$ in this non-i.i.d.\ setting, matching the original secretary problem,
and this is tight up to lower-order terms.
Our algorithm uses a single fixed threshold rule, and therefore
applies even if the values are revealed in an adaptively adversarial order.  
The proof is elementary; 
what is perhaps most surprising is that the single-threshold analysis is tight.
%On the other hand, 
Our
lower bound holds even if the order is known in advance and applies to arbitrary algorithms, 
showing that a simple fixed-threshold rule is asymptotically optimal.
This expands the long-standing
result for the i.i.d.\ case due to Gilbert and Mosteller~\cite{GM66}
to the setting with different distributions.

%As previously mentioned, in the prophet setting we assume
%that the arrival order is chosen by an adversary. It is
%natural in many settings, however, that an adversary would
%not have such power.  If so, can we design better algorithms? 
%The \emph{prophet secretary} model introduced in
%\cite{esfandiari2015prophet} is a natural way to consider such a
%process; this model follows the prophet inequality model,
%except that the adversary doesn't control the arrival order,
%and instead the observations occur in an order chosen uniformly
%at random from all permutations and revealed online.
%Again, previous work has focused on the problem of maximizing
%the expectation of the selected number, and not the probability
%of selecting the maximum number, which will be our focus here.
We next consider a random-order model, where the values are 
drawn from arbitrary independent distributions but are
presented in a uniformly random order.
The  i.i.d.\ setting of Gilbert and
Mosteller~\cite{GM66} is a special case of 
%the prophet secretary setting, 
this random order setting,
where all observed values are
chosen from the same distribution.  Our second result generalizes
the result of Gilbert and Mosteller
%~\cite{GM66} 
to show that in the
%prophet secretary 
random-order
setting, it is possible to select the maximum
value with probability at least $0.517$, using a single-threshold
algorithm.  This improves on the
adversarial-order setting, and matches the tight bound for
single-threshold algorithms for the i.i.d.\ case~\cite{GM66}.
%{\bf MM:  change large market assumption below;  good place to add 1-2 sentences 
%of motivation.}

%\textbf{[BL: Changed the large market assumption, added some discussion, and motivated the large-market
%result as a step toward showing that a multi-threshold algorithm beats the single-threshold algorithm unconditionally.]}
Still in the random-order model, we next present an algorithm that breaks this barrier
of $0.517$ using multiple thresholds.  As a corollary, algorithms that
use a single threshold are not optimal in the random-order model.
Our approach is to consider a natural ``no-superstars'' condition, 
which is that no single distribution has more than a certain constant probability (ex ante)
of generating the maximum value.  
This captures scenarios where no single entry has a non-vanishing impact on the
problem's solution in the limit as the problem size $n$ grows large.\footnote{We can think of
this as a \emph{large market} condition.}
%\footnote{This definition is related to the
%notion in economics of a \emph{large market} (as opposed to a small market, where
%individuals can have non-trivial power).}
%; this captures the
%intuitive notion in economics of a large market (as opposed to a small market, where
%individuals can have non-trivial power).  
We show that under such an assumption, 
there is an algorithm that succeeds with
probability arbitrarily close to $0.5801$, the tight success probability
obtainable in the i.i.d.\ setting with multiple thresholds, in the
limit as $n$ grows large.  If the no-superstars assumption
is violated, then the presence of a highly dominant distribution again makes it
possible to improve over the single-threshold bound of $0.517$.
%{\bf BL: someone check the wording above for accuracy.}
%we can achieve the same thresholds
%they achieve for selecting the maximum in the i.i.d. setting
%to the setting where the distributions can be different, subject
%to some small additional technical assumptions (namely, that
%the probability that any observation a priori is the maximum is
%negligible).

%{\bf BL: Moved this above the paragraph about Le Cam's theorem.}
%An algorithm that selects the maximum value with probability $\alpha$ also achieves, in expectation, at least an $\alpha$ fraction of the expected ex post maximum value.
%\textbf{[BL: Moved the connection to prophet inequalities down here.]}
It is natural to compare these results with the {\em prophet inequality}, introduced by Krengel et al.~\cite{krengel1978semiamarts,krengel1977semiamarts}. In the prophet inequality problem, the goal is to maximize the expected value of the number selected rather than the probability of selecting the maximum.  The classic prophet inequality is that one can achieve half of the expected maximum value using a single threshold algorithm, and this is tight.  The \emph{prophet secretary} model \cite{esfandiari2015prophet} considers this goal of maximizing expected value in the random-order model, which admits improved results. %and they show that the random-order model leads to improved results.  
One can view our results as extending classic ``secretary-style'' results for best-choice problems to settings typical of prophet inequalities, with independent but non-identical distributions.

%\textbf{[BL: Changed the robustness pitch of the following paragraph.]}
One distinction between the best-choice problem and a prophet inequality is that, in the best-choice problem, it is typically better to avoid a ``non-robust'' solution that achieves high expected value by accepting a very large number with very small probability, but otherwise does not obtain much value.  Motivated by this connection to robustness, one might relax the desideratum of picking only the highest number, and aim instead to obtain one of the top few values with high probability.  To this end, we consider a variant of our problem where the goal is to maximize the probability of selecting any of the top $k$ values.  A similar variant has been studied for the secretary problem
%(where items arrive in random order and only their relative ranks are revealed) 
by Gusein-Zade~\cite{GZ66}, who shows that there is an algorithm whose failure probability is at most $O(\tfrac{\log k}{k})$.  When values are drawn i.i.d.\ from a known distribution, Gilbert and Mosteller~\cite{GM66} study the case $k=2$ and solve for the limiting probability of success.  We extend this to arbitrary $k$ and arbitrary distributions presented in an adversarial order, and show that there is an algorithm with failure probability exponentially small in $k$.  Moreover, this is the best possible bound, up to coefficients in the exponent, even in the i.i.d.\ setting.

As one of our main tools in our analysis, we use Le Cam's
theorem \cite{le1960approximation}, which (as we describe below) connects sums of Bernoulli
random variables and discrete Poisson distributions.  This result, along
with coupling techniques and other additions for our setting,
allow us to represent the probability distribution  for the maximum
(over several different distributions) by discrete Poisson distributions.
This variation of a ``Poissonization'' argument for these settings
appears novel, and may be of its own interest.

%\textbf{[BL: Can we also pitch the computational aspect of our results?  I.e., that we solve in closed form, so one can compute our thresholds.  This might counter the ``why is this TCS?'' criticism we got from FOCS.]}

%\noindent \paragraph{Connection to Robustness:} { \bf BRENDAN would
%you please add the connection and the recent related works here}

\subsection{Results and Techniques}

In what follows, we refer to the \our prophet inequality problem and \our prophet secretary problem for the variations we consider, where the goal is to maximize the probability of choosing the highest observed value given distributions presented in adversarial and random order, respectively.  We start by obtaining a tight bound for the \our prophet inequalities problem:
we provide an algorithm that selects the maximum with probability at least $\frac{1}{e}$ and show that there is no algorithm that selects the maximum with probability at least $\frac{1}{e}+\epsilon$ for any constant $\epsilon > 0$.  Although the probability of success here is the same as for the classical secretary problem, the proof and corresponding algorithm are not the same.
%the approach to provide the result, and correspondingly the reason $\frac{1}{e}$ is the right probability of success, appears very different.  
Our algorithm is based on choosing a suitable threshold and accepting any observation above that threshold.  We choose the threshold to optimize the probability that exactly one element lies above it, since we are guaranteed to accept the largest value in this case.  Perhaps surprisingly, our lower bound shows that this analysis is tight, even for an arbitrary selection rule with advance knowledge of the arrival order.  %Our analysis bounds the likelihood that exactly one element lies above the threshold, which is chosen to optimize this probability.  %of multiple observations being above the threshold by coupling the process with an appropriate i.i.d. instance.

% Indeed, the nature of having a fixed order of different distributions is very different from the [almost] symmetric input of the secretary problem.
% To prove this theorem we decomposed each distribution into several smaller distributions and provide a coupling between the number of items we observe above a threshold.

We next provide a single threshold algorithm for the \our prophet secretary problem that selects the maximum with probability at least $0.517$.  This result utilizes some of the technology used for the \our prophet inequality result.
We also extend our analysis to the top-$k$-choice prophet inequality problem, and provide a single threshold algorithm that selects one of the top $k$ values with probability at least $1 - e^{-c_1 k}$, where $c_1>0$ is a fixed constant. We also show that this exponential dependence on $k$ is tight even in the i.i.d.\ setting: there is a constant $c_2>0$ such that no algorithm can select one of the top $k$ values with probability greater than $1 - e^{-c_2 k}$.  This tightness result involves arguing that an arbitrary algorithm must become ``trapped'' at some point in the observation sequence, with at least an exponential probability; conditional on what it has seen, there is a non-negligible chance that all of the top $k$ values have already been observed, but also a non-negligible chance that all of them are yet to come.
%We also present the previously mentioned results for the setting where we wish to choose one of the top $k$ values.

All of the algorithms above use a single fixed threshold.  For the \our
prophet inequality problem our lower bound shows that single-threshold algorithms
achieve tight results, but for the
\our prophet secretary problem 
we show that this is not the case.
%this does not seem to be true.
Designing and analyzing multiple-threshold algorithms is significantly
more challenging, as dependencies and correlations naturally arise. To
overcome this, we develop an alternative
approach for analyzing the setting of multiple distributions in a
random order. The intuition is that for a large number of observations $n$,
we can split the observations into consecutive groups of size $n/T$ for
a suitable constant $T$, such that we can think of the maximum of each
group as being approximately from an i.i.d.\ distribution corresponding to a sample of $n/T$ distributions from the $n$ overall distributions.
That is, each group of $n/T$ distributions is sufficiently similar that
we can view the problem as very similar to the best-choice problem for $T$ i.i.d.\ observations.
Formalizing this closeness allows us to nearly achieve the same worst-case performance of 
a $T$-threshold scheme in the i.i.d.\ setting. 
This result requires a technical 
%large-market 
``no-superstars''
condition, which is that the \emph{a priori}
probability of any specific distribution being the maximum is $o(1)$.
%$O(1/\log{n})$.  
%{\bf MM:  be sure to change this}
Using this technique, and under this no-superstars assumption, we design a threshold-based algorithm whose
success probability converges
to $0.5801$ as $n$ grows large, which is tight even for the i.i.d.\ setting.  On the other hand, 
we show that if the no-superstars assumption is violated and there exists a distribution that has more than a certain constant probability of generating
the maximum value as $n$ grows large, then one can improve the single-threshold analysis.
Combining these methods leads to an unconditional improvement over the optimal worst-case bound
for single-threshold algorithms.

\begin{comment}
 Indeed we can beat factor 0.5 slighting by setting a single
threshold.  We  try to represent the joint probability distribution
for the maximum via lots of discrete Poissons and then try to use
the standard basic theorem for discrete Poisson which is Le Cam's
Theorem~\footnote{see
\url{https://en.wikipedia.org/wiki/Le_Cam\%27s_theorem}}. However
there may be some conditions we need on the distributions especially
when we have small number of variables that are the only ones that
contribute to being over whatever threshold we set, in which case we
cannot use something like Le Cam's theorem;  perhaps that is handled
separately.  But if we are in a setting where for example every
distribution has some suitable chance to be the maximum, we believe
we can get it to work. As mentioned this is just an initial step and
we hope to get much better bounds than 0.5 using a fixed single
threshold. We are thinking about multiple thresholds as well to get
even improved bounds but that needs more new ideas to handle
correlations.%
\end{comment}

We briefly note that, unlike the expectation version of prophet
inequalities~\cite{esfandiari2015prophet}, in this setting of \our
prophet inequalities, all our results trivially extend to the setting
where we want to maximize the probability of finding the minimum
element as well.

\subsection{Poissonization Technique}
% {\bf HOSSEIN AND MICHAEL: would you please add  more details about
% our techniques as well more background and related work for Lethe
% connection and the recent related works here}

% In fact, a sequence of numbers from an iid distribution is very well behaved due to its symmetricity. In particular, for a specific threshold $\tau$ the number of draws that appear to be above $\tau$ is following from a Poisson distribution. Using this fact, by a simple analysis we show a $0.517$ approximation algorithm for \our prophet inequalities with iid distributions.
%
% To extend this result to different distributions
%
% as one of our main tools in our involved analysis, we use Le Cam's
% Theorem as well as our own parallel tools to extend the use of this theorem.
% \begin{comment}% What does the following means?
%  of it to handle discrete
% Poissons when we represent the joint probability distribution (of
% different distributions) for the maximum via lots of discrete
% Poissons, which is of its own interest.
% \end{comment}

%Le Cam's Theorem~\footnote{see
%\url{https://en.wikipedia.org/wiki/Le_Cam\%27s_theorem}}:
%{\bf
% MICHAEL: would you please add the best statement for Le Cam's
% Theorem; below are from Wikipedia more details about our techniques
% as well more background and related work for Lethe connection and
% the recent related works here}
%
% {\bf H: I moved the Wikipedia to comment and moved the statement of la cam theorem that we had in the technical to here. Probably it is still good to add some explanations}

%To Extend this result to \our prophet secretary setting, we use Le Cam's theorem defined as follow.

One approach used in \cite{GM66} involves setting a threshold and
considering the number of observations above that threshold.  In the
case of i.i.d.\ distributions for the observations, this number is the
sum of i.i.d.\ Bernoulli random variables, which is known to converge
to a Poisson distribution in the setting we consider (where the expected
number of positive observations is constant as the number of observations
grows large).

A helpful tool in extending such results to the setting where distributions
may differ for observations is Le Cam's theorem \cite{le1960approximation}.  The basic statement
of Le Cam's Theorem is the following:  let $X_1,\cdots, X_n$ be a sequence of Bernoulli random variables where $\prob{X_i=1}=p_i$ and $\lambda=\sum_{i=1}^N p_i$. We have
    $$\sum_{k=0}^{\infty} \Bigg| \prob{\sum_{i=1}^{n}\X_i = k} - \frac{\lambda^k e^{-\lambda}}{k!} \Bigg| < 2 \sum_{i=1}^{N} {p_i}^2.$$
Intuitively, Le Cam's Theorem says that when the probability of each random variable being $1$ in a sequence of Bernoulli random variables is
sufficiently small (e.g. $O(\frac 1 {N})$), the sum is well approximated by a Poisson distribution.
There are a number of interesting proofs of Le Cam's Theorem (see the survey \cite{steele1994cam}), including proofs that slightly improve the constant on the right hand
side above, but this general bound suffices for our purposes.

\begin{comment}
 In
probability theory, Le Cam's theorem, named after Lucien le Cam,
states the following.[1?][2?][3?]

Suppose:

$X_1,\cdots, X_n$ are independent random variables, each with a
Bernoulli distribution (i.e., equal to either 0 or 1), not
necessarily identically distributed. $Pr(X_i = 1) = p_i$ for $i = 1,
2, 3, cdots$ $ \lambda _{n}=p_{1}+\cdots +p_{n}.$ $
\lambda_{n}=p_{1}+\cdots +p_{n}.$ $ S_{n}=X_{1}+\cdots +X_{n}.$ $
S_n = X_1 + \cdots + X_n$. (i.e.  $S_{n}$ follows a Poisson binomial
distribution) Then

$ \sum _{k=0}^{\infty }\left|\Pr(S_{n}=k)-{\lambda
_{n}^{k}e^{-\lambda _{n}} \over k!}\right|<2\sum
_{i=1}^{n}p_{i}^{2}.$ $\sum _{{k=0}}^{\infty
}\left|\Pr(S_{n}=k)-{\lambda _{n}^{k}e^{{-\lambda _{n}}} \over
k!}\right|<2\sum _{{i=1}}^{n}p_{i}^{2}$. In other words, the sum has
approximately a Poisson distribution and the above inequality bounds
the approximation error in terms of the total variation distance.

By setting $pi = ?n/n$, we see that this generalizes the usual
Poisson limit theorem.

When $ \lambda _{n}$ $\lambda _{n}$ is large a better bound is
possible: $ \sum _{k=0}^{\infty }\left|\Pr(S_{n}=k)-{\lambda
_{n}^{k}e^{-\lambda _{n}} \over k!}\right|<2(1\wedge {\frac
{1}{\lambda }}_{n})\sum _{i=1}^{n}p_{i}^{2}.$ $ \sum _{k=0}^{\infty
}\left|\Pr(S_{n}=k)-{\lambda _{n}^{k}e^{-\lambda _{n}} \over
k!}\right|<2(1\wedge {\frac {1}{\lambda }}_{n})\sum
_{i=1}^{n}p_{i}^{2}.$ [4?]

It is also possible to weaken the independence requirement.[4?]
\end{comment}

\subsection{Outline of Paper}
After formalizing notation in Section~\ref{sec:notation}, we present
our bounds for the \our prophet inequality problem and \our prophet
secretary problem in Section~\ref{sec:SingleT}.  In
Section~\ref{sec:TopK} we generalize these results to the problem of
selecting any of the top $k$ values.  Section~\ref{sec:iidRed}
contains our most technically demanding result, which is that the
bound achievable for the \our prophet secretary problem improves to
the (known) bound for the i.i.d. case using several thresholds. Then
in Section~\ref{sec:related} we summarize additional related work,
and we conclude with some open problems in
Section~\ref{sec:conclusion}.

\section{Notation}
\label{sec:notation}

In the \emph{\our prophet inequality problem}, we are given
a set of distributions $\{D_1,\ldots,D_n\}$.  We then observe an online 
sequence of values $x_1 ,\cdots, x_n$, where each $x_i$ is drawn independently from $D_i$,
presented in an arbitrary order.  When value $x_i$ is observed, we must
irrevocably decide whether or not to choose that value. Once we choose 
a value, the process stops.  A value that has been observed but not 
chosen cannot be chosen later.
The goal is to maximize the probability that the value chosen is equal to 
$\max_i \{x_i\}$.
We emphasize that the order in which the values are presented is arbitrary and 
not known in
advance.  We refer to the case where the distributions are identical as the \emph{i.i.d.\ setting}.

The \emph{\our prophet secretary problem} is identical, except that the
values are presented in a uniformly random order.  That is, 
%
%In the \emph{\our prophet secretary} problem we are similarly given 
%a set $\{D_1,\ldots,D_n\}$ of distributions. 
%A value $x_i$ is drawn from each distribution $D_i$. 
after applying a random permutation $\Pi=\langle \pi_1,\ldots, \pi_n \rangle$ 
on the sequence of $x_i$ values, they are presented online in that order, so that
at step $k$, $\pi_k$ and $x_{\pi_k}$ are revealed.  Again, the goal is to maximize
the probability of choosing a maximum value.

Our algorithms will be threshold-based, where we choose a value
if and only if it lies above a suitable threshold.
%{\bf MM:  I don't get what weakly is doing there, can we remove}
We use $\Tau = \langle \tau_1,\ldots, \tau_n \rangle$ to refer to a
sequence of thresholds;  thus, we check for example whether $x_{\pi_k} \geq \tau_k$.
In the case that $\tau_1 = \tau_2 = \dotsc = \tau_n = \tau$, we say that
the algorithm is a \emph{single-threshold algorithm}.
% BL: Dropped the $p_i$ notation; we redefine it when we use it, and we were only using it for fixed-threshold proofs anyway.
%, and we let $p_i^{\Tau} =
%\prob{x_{i}>\tau_{\pi_i}}$.  Note that $p_i^{\Tau}$ depends on $\Pi$, but
%we supress this dependency for simplicity. We also drop $\Tau$ from the notation 
%and simply write $p_i$ when it is clear from the context.

In our proofs, we will assume for notational convenience that the distributions are \emph{atomless:}
the probability distributions are continuous, so that no single value takes on a non-zero probability.  
We use this assumption only to define the inverse of a given cumulative distribution; i.e.,
to find a value $\tau$ such that $\Pr_{x \sim D}[x \geq \tau] = p$ for some fixed $p \in [0,1]$.
This is only for convenience, and our results actually apply to the general case with atoms,
using the following reduction based on using an auxiliary random number to break ties (which we believe is folklore).
If there exists a value $\tau$ such that 
$\Pr_{x \sim D}[x \geq \tau] > p$ but also $\Pr_{x \sim D}[x \leq \tau] \geq p$ (i.e., there is an atom that
prevents the desired inversion), then we can modify our random
process to include a random variable $y$ drawn from the uniform distribution on $[0,1]$, and augment
threshold $\tau$ with a secondary threshold $\bar{y}$.  We will then interpret the event $[x \geq \tau]$ 
to mean $[(x > \tau) \vee ((x = \tau) \wedge (y \geq \bar{y}))]$, and set $y$ so that, under this
definition, $\Pr_{x \sim D}[x \leq \tau] = p$.  With this reduction in mind, we will assume throughout that distributions
are atomless without further comment.

%Whenever drawing 
%This assumption is used to 
%We can assume this without loss of generality for our theorems, using the following reduction (which
%we believe is folklore).  

\section{Best-Choice Algorithms with a Single threshold}
\label{sec:SingleT}
In this section, we describe algorithms and lower bounds for the \our prophet inequality problem (in Section~\ref{subsec:prophet}) and the \our prophet secretary problem (in Section~\ref{subsec:PS}).  All of the algorithms in this section will be single-threshold algorithms.
%, i.e., $\tau_1=\tau_2=\dots=\tau_n = \tau$, first for \our prophet inequalities, and then for \our prophet secretary.

\subsection{Best-Choice Prophet Inequalities} \label{sec:pi}
We begin by showing that it is possible to choose the maximum value with probability at least $\frac{1}{e}$, using a single threshold, for the \our prophet inequality problem.  
%We start with a single threshold algorithm with a simple analysis that chooses the maximum with probability at least $\frac{1}{4}$ for \our prophet inequalities;  this result shows a basic framework we utilize for other results.  In particular, we expand on this approach to 
%improve our result to give an algorithm that chooses the maximum with probability at least $\frac{1}{e}$ using a single threshold in Theorem \ref{thm:pi:alg:improved}. In Theorem~\ref{thm:pi:hard}, we show that this result is tight.

%\begin{theorem}\label{thm:pi:alg}
\begin{theorem}\label{thm:pi:alg:improved}
For the \our prophet inequality problem, there is an algorithm that succeeds with probability at least $\frac{1}{e}$.
\end{theorem}
\begin{proof}
We will warm up by proving an easier result: a simple single-threshold algorithm that succeeds with probability $1/4$. We'll then show how to improve this to $1/e$.
Our algorithm will select threshold $\tau$ such that $\prob{\maxall \geq \tau}=1/2$, and choose the first value that is at least $\tau$.  From the definition of $\tau$, the algorithm chooses a value with probability $1/2$, otherwise it chooses nothing.  Conditional on having chosen a value, the algorithm will certainly succeed if no subsequent value is strictly greater than $\tau$.  But the probability of a subsequent value lying above $\tau$ is at most $1/2$, the probability that \emph{any} of the $n$ observations is greater than $\tau$.  So the probability of success, conditional on having selected an item, is at least $1/2$, leading to a total success probability of at least $1/4$.\footnote{This warm up is similar to ~\cite{samuel1984comparison}.}

We can modify the algorithm above to improve the success probability to $1/e$.  Namely, the algorithm will set threshold $\tau$ so that $\prob{\maxall \leq \tau}=1/e$ and pick the first number that is larger than $\tau$. We show that with probability at least $1/e$ there is exactly one number which is larger than $\tau$, which implies the desired result.  Let $p_i = \prob{x_{i}>\tau}$. By the way we choose $\tau$, we have 
\begin{align}\label{eq:pi:alg}
\prod_{i=1}^n (1-p_i)=1/e. 
\end{align}
We now consider the probability that exactly one number is larger than $\tau$, and show that it is at least $1/e$;
this completes the proof.\footnote{We thank an anonymous reviewer for providing us this simplification over our prior proof.}  
The probability that the $j$th observed value is larger than $\tau$ but all others are not is
\begin{align}\label{eq:pi2:alg}
\frac{p_j}{1-p_j} \prod_{i=1}^n (1-p_i) = \frac{1}{e} \frac{p_j}{1-p_j}.
\end{align}
We briefly note the fact that $e^x \geq 1+x$ implies (using $x= p_j/(1-p_j)$)
\begin{align}\label{eq:pi3:alg}
\frac{p_j}{1-p_j} \geq \ln \left ( \frac{1}{1-p_j} \right).
\end{align}
Now the probability that exactly one number is larger than $\tau$ is 
\begin{align*}
\frac{1}{e} \sum_{j=1}^n \frac{p_j}{1-p_j} & \geq  \frac{1}{e} \sum_{j=1}^n \ln \left ( \frac{1}{1-p_j} \right) \\
& =  \frac{1}{e} \ln \frac{1}{\prod_{j=1}^n  (1-p_j)} \\
& =  \frac{1}{e}.
\end{align*}
Here the first line follows from Inequality~\ref{eq:pi3:alg}, and the last line from 
$\prod_{j=1}^n  (1-p_j) = 1/e$.  
\end{proof}

Our algorithm uses only a single fixed threshold as its stopping rule.  One might suspect that a more complicated algorithm, perhaps one that modifies its thresholds adaptively or employs randomization, would perform better.  Our next result is that this is not the case: no online algorithm can guarantee a success probability strictly better than $\frac{1}{e}$.

\begin{theorem}\label{thm:pi:hard}
	For any constant $\eps > 0$, there is no algorithm that succeeds with probability $\frac{1}{e}+\eps$ for the \our prophet inequality problem. 
\end{theorem}
\begin{proof}
	Consider the following example. There are $n$ random variables $x_1,\dots,x_n$ from distributions $D_1,\dots,D_n$ as follows: for any $i\in \{1,\cdots,n\}$, $x_i$ is $i$ with probability $q_i=\frac 1 i$ and $0$ otherwise.\footnote{One can make this atomless by assuming that $x_i$ is drown uniformly at random from $[i,i+\epsilon]$ with probability $q_i=\frac 1 i$ and $0$ otherwise.} Note that the nonzero random variable with the largest index is the maximum. Hence the probability of $x_i$ being the maximum is independent of the $x_j$ values with $j<i$. Moreover, $x_1$ is always $1$ and hence the maximum is never $0$. We let $p_i$  be the probability that $x_i$ is the maximum. The distributions will arrive in index order.
	
	We claim that $p_1=p_2=\dots=p_n=\frac 1 n$.  We will show this using strong induction.\footnote{This follows similar reasoning to the analysis of reservoir sampling.}  The base case holds for $i=n$ where $p_n=q_n=\frac 1 n$.  Assuming $p_{i+1}=\dots=p_n=\frac 1 n$, we have 
\begin{align*}
	p_i & = q_i\cdot\prob{x_{i+1}=\dots=x_n=0}\\
	&= q_i\cdot\prob{\text{maximum is not in $\{x_{i+1},\dots,x_n\}$ }}\\
	& = q_i\Big(1-\sum_{j=i+1}^{n} p_j\Big)
	=\frac 1 i \frac i n = \frac 1 n.
	\end{align*}
	Hence we have $p_1=p_2=\dots=p_n=\frac 1 n$. Also, we have $\prob{x_i=\max_{j=1}^n x_j|x_i\neq 0}=\frac i n$.
%  and $\prob{x_i=\max_{j=1}^n x_j}=0$ otherwise. Remark that $\prob{x_i=\max_{j=1}^n x_j | x_i\neq 0}$ is increasing in $i$.
	
	Let $\alg$ be the best online algorithm and let $\alg_{i+1\rightarrow n}$ be the probability that $\alg$ picks the maximum assuming that it rejects $x_1,\dots, x_i$. 
% Without loss of generality, we assume $\alg$ never picks $zero$, although it might not pick anything by the end of the stream. 
Notice that if $\alg$ picks a nonzero number $x_j$ from $\{x_{i+1},\dots,x_n\}$, it is larger than all numbers in $\{x_1,\dots, x_i\}$. Hence $\alg_{i+1\rightarrow n}$ is independent of $x_1,\dots, x_i$. Notice that if $\alg$ rejects $x_{i+1}$ it picks the maximum with probability $\alg_{i+2\rightarrow n}$. Hence $\alg_{i+1\rightarrow n} \geq \alg_{i+2\rightarrow n}$, which means $\alg_{i+1\rightarrow n}$ is decreasing in $i$.
	
	Indeed, if $x_i\neq 0$, when $\prob{x_i=\max_{j=1}^n x_j | x_i\neq 0} \geq \alg_{i+1\rightarrow n}$, $\alg$ picks $x_i$ and stops. Otherwise, $\alg$ rejects $x_i$ and continues. Also remember that $\prob{x_i=\max_{j=1}^n x_j | x_i\neq 0}$ is increasing in $i$ and $\alg_{i+1\rightarrow n}$ is decreasing in $i$. Therefore, there exists an index $i$ such that for all $j<i$, $\alg$ rejects $x_i$ and accepts the first nonzero $x_j$ with $j\geq i$. Therefore, $\alg$ picks the maximum with probability 
	\begin{align*}
		\sum_{j=i}^{n} \prob{x_i=\dots=x_{j-1}=0}\prob{x_j=\max_{j'=1}^n x_{j'}} &=\frac 1 n + \sum_{j=i+1}^{n} \Pi_{k=i}^{j-1}(1-1/k) \frac 1 n \\
		&= \frac 1 n + \sum_{j=i+1}^{n} \frac{i-1}{j-1} \frac 1 n \\
		&= \frac 1 n + \frac {i-1}{n} \sum_{j=i+1}^{n} \frac 1 {j-1} \\
		&\leq \frac 1 n + \frac {i-1}{n} \Big( \ln (\frac {n+1}{i}) + \frac{1}{i} \Big)\\
		&\leq \frac 2 n + \frac {i-1}{n}  \ln (\frac {n+1}{i}) \\ 
		&\leq \frac 2 n + \frac {i}{n+1}  \ln (\frac {n+1}{i}) \\
		&\leq \frac 2 n + \alpha \ln (1/\alpha) &\text{for $\alpha=\frac {i}{n+1}\in [0,1]$.} 
	\end{align*}
	Note that $\alpha \ln (1/\alpha)$ maximizes at $\alpha=\frac 1 e$. Thus, $\alg$ picks the maximum with probability at most $\frac 1 e + \frac 2 n$.
\end{proof}

\label{subsec:prophet}
\subsection{Best-Choice Prophet Secretary} \label{sec:ps}
In this subsection we show a single threshold suffices to provide an algorithm that chooses the maximum value with probability $0.517$ for \our prophet secretary. To begin, we provide a simple analysis that achieves this $0.517$ probability for \our prophet inequalities with i.i.d.\ distributions.  
We note that this result was presented in \cite{GM66}, with the constant calculated numerically for large values of $n$.  We essentially follow their argument, but provide a formal justification for their numerical results.  
%\textbf{[BL: We say that the difference between this and GM66 is that GM66 is calculated numerically, but we have a ``calculated numerically'' step in our analysis.  So it's not obvious what the difference is.  Can we say more about why our analysis is different?]}
% MM:  I think this is fine -- we actually just do the asymptotics, so we do just what we say.  

\begin{theorem}\label{thm:iid:alg1}
	For sufficiently large $n$, there exists a single threshold algorithm that chooses the maximum value with probability arbitrarily close to $\max_\lambda \sum_{k=1}^{\infty}\Big( \frac 1 k \frac{\lambda^k e^{-\lambda}}{k!} \Big) \approx 0.5173$, for \our prophet inequalities with i.i.d. distributions, and this is tight for single-threshold algorithms.
\end{theorem}
\begin{proof}
	Let $\tau$ be given by $\prob{\maxall \leq \tau}=\P$, and $p = \prob{x_i\geq \tau} = 1-\P^{1/n}$ for $\P$ to be given later.  Let $\mathcal{K}$ be the random variable indicating the number of $x_i$ that are greater than $\tau$. When $\mathcal{K}\geq 1$, due to symmetry each of these $\mathcal{K}$ items is the maximum with probability $1/\mathcal{K}$, and since we pick the first item that is greater than $\tau$, when $\mathcal{K}\geq 1$ the maximum is chosen with probability $1/\mathcal{K}$. Thus, the probability that we pick the maximum is at least 
	\begin{align*}
		\sum_{k=1}^{n}\Big( \frac 1 k \prob{\mathcal{K}=k} \Big).
	\end{align*}
Here $\mathcal{K}$ is sum of Bernoulli random variables, and so the probability we choose the maximum is simply
	\begin{align*}
		\sum_{k=1}^{n}\Big( \frac 1 k {n \choose k} p^k (1-p)^{n-k} \Big).
	\end{align*}
For large $n$ we may use that the limit of the Bernoulli distribution becomes a Poisson distribution, and use numerical
calculations and Le Cam's theorem to obtain the result.  Specifically, take 
$\P=(1-1.501/n)^n \simeq e^{-1.501}$ and $p=\prob{x_i\geq \tau}=1.501/n$, where the $1.501$ is determined numerically.
By Le Cam's theorem $\sum_{k=0}^{\infty} \Bigg| \prob{\mathcal{K}=k} - \frac{\lambda^k e^{-\lambda}}{k!} \Bigg| < 2 n p^2$, where $\lambda=n p = 1.501$. This gives us $\sum_{k=1}^{\infty} \Bigg| \frac 1 k \prob{\mathcal{K}=k} - \frac 1 k \frac{\lambda^k e^{-\lambda}}{k!} \Bigg| < 2 n p^2 < \frac 6 n$. Therefore the probability that we pick the maximum is at least
	\begin{align*}
		\sum_{k=1}^{n}\Big( \frac 1 k \prob{\mathcal{K}=k} \Big) &\geq \sum_{k=1}^{n}\Big( \frac 1 k \frac{\lambda^k e^{-\lambda}}{k!} \Big) - \frac 6 n \\
		&\geq 0.5173 - \frac 6 n &\text{Calculated numerically for $\lambda = 1.501$}\\
		& \geq 0.517. &\text{Assuming $n\geq 20000$}
	\end{align*}
We note that by taking $n$ large enough, we can obtain a success probability arbitrarily close to the sum
$$\max_\lambda \sum_{k=1}^{\infty}\Big( \frac 1 k \frac{\lambda^k e^{-\lambda}}{k!} \Big)$$
using the same argument.  This is an asymptotic upper bound by a similar argument, so this success probability is tight.
\end{proof}

\begin{comment} %MOVED TO THE INTRO
To extend this result to \our prophet secretary setting, we use Le Cam's theorem defined as follow.

\begin{theorem}[Le Cam's Theorem]\label{thm:LeCam}
	Let $\chi_1,\dots,\chi_N$ be a sequence of Bernoulli random variables where $\prob{\chi_i=1}=p'_i$ and $\lambda=\sum_{i\in [1,N]} p'_i$. We have
	$$\sum_{k=0}^{\infty} \Bigg| \prob{\sum_{i=1}^{n}\chi_i = k} - \frac{\lambda^k e^{-\lambda}}{k!} \Bigg| < 2 \sum_{i=1}^{N} {p'_i}^2.$$
\end{theorem}
\end{comment}

We are now ready to extend Theorem~\ref{thm:iid:alg1} to the more general \our prophet secretary problem.  Notice that the following theorem does not require $n$ to be large, so even when applied to the special case of i.i.d.\ distributions it extends Theorem~\ref{thm:iid:alg1} to general $n$.  

\begin{theorem}\label{thm:ps:alg1}
	For any $\eps' > 0$, there exists a single threshold algorithm that chooses the maximum value with probability at least $\max_\lambda \sum_{k=1}^{\infty}\Big( \frac 1 k \frac{\lambda^k e^{-\lambda}}{k!} \Big) \approx 0.5173 $, for the \our prophet secretary problem.
\end{theorem}
\begin{proof}
	As in Theorem~\ref{thm:iid:alg1}, we set $\tau$ such that $\prob{\maxall \leq \tau} \simeq e^{-1.501}$ and pick the first number which is at least $\tau$. We clarify the exact value of $\tau$ later in the proof after we present the required notation.
	To analyze the algorithm, for some arbitrary small $\eps'$ we replace each distribution $D_i$ with a bag of $n^2/\eps'$ identical and independent copies of a dummy distributions $D'_i$, where the distribution of the maximum of the $n^2/\eps'$ copies of $D'_i$ is equivalent to $D_i$. We let $x_i^j$ to be the realization of the $j$'th copy of $D'_i$, let $p_i^j = \prob{x_i^j \geq \tau}$, and let $n'=n^3/\eps'$ to be the total number of dummy distributions. By the way we have defined the dummy distributions, the distribution of the maximum of all dummy distributions is equivalent to the distribution of the maximum of the original problem.  
	
	The bags arrive in a random order and upon the arrival of each bag we observe the realization of the maximum number in the bag. The first time we face a bag with at least one number above the threshold, we stop and pick the maximum number in the bag. Again, the distribution of the value chosen in this framework is equivalent to that of our threshold algorithm on the actual distributions. 
	
	Let $\mathcal{K}$ be the random variable indicating the number of $x_i$s that are greater than $\tau$ and let $\mathcal{K}'$ be the random variable indicating the number of $x_i^j$s that are greater than $\tau$. In fact, if for some $i$ we have $x_i\geq \tau$, then for some $j$ we have $x_i^j \geq \tau$. Hence we have $\mathcal{K}' \geq \mathcal{K}$. 
	Notice that if $\mathcal{K}'\geq 1$ with probability $1/\mathcal{K}$ the bag that contains the maximum number arrives first and we select the maximum number; otherwise, we do not. Thus, we choose the maximum with probability 
	\begin{align}\label{eq:ps:bound}
	\sum_{k=1}^{n'}\Big( \prob{\text{We choose the maximum} \big| \mathcal{K}'=k} \prob{\mathcal{K}'=k} \Big) \geq 
	\sum_{k=1}^{n'}\Big( \frac 1 k \prob{\mathcal{K}'=k} \Big), 
	\end{align}
	where the inequality holds since $\mathcal{K}' \geq \mathcal{K}$.
	
	%XXX  		$\mathcal{K}$ does not follow poisson but $\mathcal{K}'$ does XXX.
	Now we are ready to set the value for $\tau$ given at the beginning of the proof;  specifically,  we set $\tau$ so that $\sum_{i=1}^{n}\sum_{j=1}^{n^2/\eps'} p_i^j$ equals the value of $\lambda$ that maximizes
$\max_\lambda \sum_{k=1}^{\infty}\Big( \frac 1 k \frac{\lambda^k e^{-\lambda}}{k!} \Big)$, which is approximately $1.501$.  This corresponds to $\lambda = 1.501$ for Le Cam's Theorem.  Also for any $i$ and $j$ we have $p_i^j\leq \eps'/n^2$. Using Le Cam's Theorem we have 
	\begin{align*}
			\sum_{k=1}^{\infty} \Bigg| \frac 1 k \prob{\mathcal{K}' = k} - \frac 1 k  \frac{\lambda^k e^{-\lambda}}{k!} \Bigg| &\leq \sum_{k=1}^{\infty}\Bigg| \prob{\mathcal{K}' = k} - \frac{\lambda^k e^{-\lambda}}{k!} \Bigg| \\
		&< 2 \sum_{i=1}^{n}\sum_{j=1}^{n^2/\eps'} {p_i^j}^2  &\text{Le Cam's Theorem}\\
		&\leq 2 \frac{n^3}{\eps'} \times \big(\frac {\eps'}{n^2}\big)^2 \leq \eps' &\text{since $p_i^j\leq \frac{\eps'}{n^2}$ and $n>1$.}
	\end{align*}
	This immediately gives us $\sum_{k=1}^{\infty}\Big( \frac 1 k \prob{\mathcal{K}'=k} \Big) \geq \sum_{k=1}^{\infty} \Big( \frac 1 k  \frac{\lambda^k e^{-\lambda}}{k!}\Big) - \eps'$. Therefore, the probability that our algorithm picks the maximum is at least 
	\begin{align*}
		\sum_{k=1}^{n'}\Big( \frac 1 k \prob{\mathcal{K}'=k} \Big) &= \sum_{k=1}^{\infty}\Big( \frac 1 k \prob{\mathcal{K}'=k} \Big) &\text{Since for $k>n'$, $\prob{\mathcal{K}'=k}=0$}\\
		&\geq \sum_{k=1}^{\infty} \Big( \frac 1 k  \frac{\lambda^k e^{-\lambda}}{k!}\Big) - \eps'.
% 		&\geq 0.5173 - \eps &\text{Calculated numerically}\\
%		& \geq 0.517. &\text{Setting $\eps \leq 0.0003$}
	\end{align*}
	Recall that $\eps'$ is an arbitrary small positive number and the algorithm does not depend on $\eps'$. Hence, the probability that our algorithm picks the maximum is at least $\sum_{k=1}^{\infty} \Big( \frac 1 k  \frac{\lambda^k e^{-\lambda}}{k!}\Big)$ as claimed.
\end{proof}

We note that since the lower bound in Theorem~\ref{thm:ps:alg1} matches the upper bound on the performance of any single-threshold algorithm from Theorem~\ref{thm:iid:alg1}, we can conclude that the algorithm in Theorem~\ref{thm:ps:alg1} is best-possible among single-threshold algorithms for best-choice prophet secretary.

\label{subsec:PS}

\section{Top-$k$-Choice Algorithms}
\label{sec:TopK}

In this section we consider a variant of our best-choice problems, where the goal is relaxed to choosing one of the $k$ largest values.  Here $k > 1$ is fixed as $n$ grows large.  As before, we can make only a single choice; doing so stops the process, and that is the final selection.  We first show that for the top-$k$-choice prophet inequality problem, where the distributions are presented in an arbitrary order, there is a single-threshold algorithm whose probability of failure is exponentially small in $k$.
%, and no algorithms can do better than haveing exponentially small probability of failure in $k$.

\begin{theorem}
\label{thm:topk}
For any $k \geq 1$, there exists an algorithm for the top-$k$-choice prophet inequality problem that succeeds with probability at least $1 - 2e^{-\gamma k}$, where $\gamma = (3-\sqrt{5})/2$. 
\end{theorem}
%\begin{proof}
%For a given constant $t$, let $X(t)$ be the random variable corresponding to the number of items $i$ such that $x_i \geq t$.  Choose $\tau$ so that $E[X(\tau)] = k/2$.  (We assume atomless distributions here for simplicity.)  
%
%The single threshold algorithm with threshold $\tau$ will succeed unless $X(\tau) = 0$ or $X(\tau) > k$.  We note that $X(\tau)$ is the sum of $n$ Bernoulli random variables, where variable $i$ is $1$ with probability $\Pr[ x_i \geq t ]$.  By the additive form of the Chernoff bound, we have that
%\[ \Pr[ X(\tau) = 0 ] = \Pr[ X(\tau) \leq \E[X(\tau)] - k/2 ] < e^{-KL( 0 || k/2n ) \cdot n } \]
%where $KL(p || q)$ denotes the KL divergence.  Using the bound $KL(p || q) \geq (p-q)^2 / q$ for $p < q$, we have that
%\[ \Pr[ X(\tau) = 0 ] < e^{-KL( 0 || k/2n ) \cdot n } < e^{ n \cdot ( k / 2n )^2 / (k / 2n) } = e^{-k/4}. \]
%Similarly, we have
%\[ \Pr[ X(\tau) > k ] = \Pr[ X(\tau) > \E[X(\tau)] - k/2 ] < e^{-KL( k/n || k/2n ) \cdot n } < e^{ n \cdot ( k / 2n )^2 / (k / n) } = e^{-k/2} \]
%where the second inequality uses the bound $KL(p || q) \geq (p-q)^2 / p$ for $p > q$.  Taking a union bound over these two events completes the proof.
%% = \inf\{ t \colon E[X(t)] < k/2\}$.  Then $E[X(\tau)] \geq k/2$, and $\tau$ is the largest such  
%%size of the set $\{ i \colon x_i \geq \tau \}$.  
%\end{proof}
%We note that if we choose a threshold $\tau$ so that $E[X(\tau)] = \gamma k$ for $\gamma = (3-\sqrt{5})/2$ we succeed
%a slightly better probability $1 - 2e^{-\gamma k}$ with the same argument;  we have not sought to optimize the constant further. 

The algorithm in Theorem~\ref{thm:topk} sets its threshold $\tau$ so that the expected number of values greater than $\tau$ is exactly $\gamma k$.  The result then follows by applying standard concentration bounds (Chernoff) to show that it is exponentially unlikely (in $k$) that no values are greater than $\tau$, and also exponentially unlikely that strictly more than $k$ values are greater than $\tau$.  The formal details are deferred to Appendix~\ref{app:TopK}.
%unlikely that either no values are greater than $\tau$ or more than $k$ values are greater than $\tau$.  

One thing to note about the bound in Theorem~\ref{thm:topk} is that it is independent of $n$, which we can take to be very large relative to $k$.  It's tempting to imagine that one could improve this error in special cases such as the i.i.d.\ setting.  Our next result shows that this is not possible.  One cannot do better than an exponentially decreasing error in $k$, even for the i.i.d.\ setting and hence also for the top-$k$-choice prophet secretary problem.  

We note that for such a bound one cannot simply condition on observing a certain worst-case ordering over a collection of $\theta(k)$ distributions, as the probability of seeing any particular permutation of $\theta(k)$ elements is $e^{-\theta(k \log k)}$.  The intuition of our proof is that, say halfway through the process, there is at least an exponentially small probability that the algorithm becomes ``trapped:'' given what it has seen, there is at least an exponentially small probability that all of the top $k$ values were present in the first half, but also at least an exponentially small probability that all of the top $k$ values appear in the second half.  Thus, regardless of what the algorithm has done, an exponential error bound cannot be avoided.
%will have been observed, in which case the algorithm will have to have selected one of them in order to succeed, but there is also at least an exponentially small probability that the top $k$ values will lie in the second half, so 
%the bound still cannot be avoided.  
Formalizing this intuition takes some care.

\begin{theorem}
\label{thm:topk.lb}
There exists a constant $c$ such that, for any fixed $k \geq 1$, no algorithm for the top-$k$-choice prophet inequality problem with identical distributions %\our prophet secretary problem 
selects the maximum with probability more than $1 - e^{-c \cdot k}$.
\end{theorem}

\section{Improved Best-Choice Prophet Secretary with Multiple Thresholds}\label{sec:iidRed} %Distributions in Random Order

As we showed in Section \ref{subsec:prophet}, a single threshold algorithm achieves tight results for \our prophet inequalities. However, this does not seem to be true for \our prophet secretary. In this section, which captures our main result, we seek to go beyond the single threshold algorithms and design a more efficient algorithm for \our prophet secretary. Our algorithm will use multiple thresholds.
%\textbf{HE: Please see the following paragraph}
First we provide an algorithm for inputs with an additional assumption that we call the \emph{no-superstars assumption}, which is that no single observation has too large a probability, \emph{a priori}, of being the largest value. Then we use this algorithm to provide an unconditional algorithm for \our prophet secretary that improves upon single threshold algorithms. 
%\textbf{[BL: I haven't changed the name ``small bid'' beyond this point.]}

%Specifically, we have the following definition:
\begin{definition}[No-Superstars Assumption.] \label{def:iid:smallbid}
	We say that a set of distributions $\{D_1,\dots,D_n\}$ satisfies the {\em no $\eps$-superstars assumption} if, for all $i\in \{1,\dots,n\}$,
	$$\prob{i = \argmax_{j=1}^n x_j}\leq \eps,$$
	where each $x_i$ is a random variable drawn from $D_i$. 
\end{definition}
In particular, we will show that our algorithm results in an improved bound (relative to the best single-threshold algorithm) when the set of distributions satisfies a no $\eps$-superstars assumption for a sufficiently small constant $\eps$. We will sometimes drop the $\eps$ and simply refer to the ``no-superstars assumption'' when $\eps$ is clear from context.

%%%%%%%%%%%%%%%%%%%%%%%%%%%%%%%%%%%%%%%%%%%

The starting point for our algorithm is the analysis of Gilbert and Mosteller~\cite{GM66}, which shows that in the i.i.d.\ setting the optimal (multi-threshold) algorithm succeeds with probability $0.5801$ as $n$ grows large.  At an intuitive level, we would like to establish that a prophet secretary instance behaves similarly to an i.i.d.\ instance, where each of the distinct distributions is replaced by an ``average'' of all the distributions.  However, this is not quite right due to correlations between values.  For example, once the process reaches the last few distributions, the algorithm may have a lot of information about their likely outcomes relative to an i.i.d.\ instance, because knowing which distributions are left could be very informative.

%what we would like to do for our proof is replace each of the distinct distributions by an ``average'' of all the distributions.  This appears difficult to do, because of dependencies.  (For example, if the process reaches the last few distributions, much may be known about their likely outcomes, because knowing the distributions that are left to be observed may contain a great deal of information.)  
To dampen this correlation, we will instead consider groups of $qn$ consecutive observations for some small constant $q$.  The maximum of each collection of $qn$ distributions will, because of concentration from sampling, be distributed very similarly to the maximum of a suitable average of all the distributions, and there is negligible correlation between the $1/q$ collections.  It is here where we make use of the no-superstars assumption.  We can therefore model our \our prophet secretary instance as a (nearly) i.i.d.\ instance with $1/q$ observations, and design an algorithm based on the i.i.d.\ variation of the problem.  This ultimately leads to an algorithm for \our prophet secretary that succeeds with probability as close as desired to the worst-case guarantee of the best i.i.d.\ algorithm.

\begin{theorem}\label{thm:Ralg:main:foreshadow}
	Let $\alg_{\tau}$ be any threshold-based algorithm that selects the maximum with probability at least $\alpha$ when values are i.i.d.  Then for any $\gamma \in (0,1)$, there is an algorithm for the \our prophet secretary problem that selects the maximum with probability at least $(\alpha-13\gamma)$, whenever the distributions satisfy the no $\eps$-superstars assumption with $\eps =\frac{\gamma^{10}}{24\log(\frac 2 {\gamma^2})}$.  In particular, we can take $\alpha \approx 0.5801$ as $n$ grows large.
\end{theorem}

%\textbf{HE: see the following}

While Theorem~\ref{thm:Ralg:main:foreshadow} requires a no-superstars assumption, we can use it to 
show that for general input distributions, the single-threshold algorithm is not tight, under the additional assumption that we observe not just the value but also which distribution the value arises from in each observation. 

\begin{theorem}\label{thm:ps:algEps}
	There exists an algorithm for the \our prophet secretary problem that chooses the maximum value with probability at least $\max_\lambda \sum_{k=1}^{\infty}\Big( \frac 1 k \frac{\lambda^k e^{-\lambda}}{k!} \Big) + \eps_0 $, where $\eps_0$ is a positive constant, when we observe not just the value but also the distribution from which each value arises.
\end{theorem}

%To capture this intuition, we instead consider groups $qn$ consecutive observations, for some small $q$.  The maximum of each collection of $qn$ distributions will, because of concentration from sampling, be distributed very similarly to the maximum of a suitable average of all the distributions.  (See Lemma \ref{lm:Ralg:DS}.  We emphasize that this is where the small bid assumption comes into play.)  It follows that we can provide an algorithm based on the i.i.d. variation of the problem, where here we have $1/q$ observations from nearly i.i.d. distributions.  This allows us in Subsection \ref{subsec:reduc-alg} to use these tools to provide an algorithm for \our prophet secretary that succeeds with probability almost that of the best i.i.d. algorithm.  We formalize these ideas more carefully below.  

We give the formal details of our algorithm and analyze its success probability in Section~\ref{subsec:reduc-alg}.  Omitted proof details appear in Appendix~\ref{subsec:reduc-alg-proofs}.  The main technical difficulty in the analysis is establishing the necessary concentration bounds, which require some care because we are sampling without replacement and do not have a good uniform bound on the contribution of any single value.  We defer the proof details of these concentration inequalities to 
%Section~\ref{subsec:reduc-conc}.
Appendix~\ref{subsec:reduc-conc}.

%%%%%%%%%%%%%%%%%%%%%%%%%%%%%%%%%%%%%%%%%%%

% BL: Subsection containing the algorithm description and analysis
\subsection{An Algorithm for Best-Choice Prophet Secretary}\label{subsec:reduc-alg}

Before we describe our algorithm for the \our prophet secretary problem,
%under an $\epsilon$-small bid assumption where $\epsilon = \theta(1 / \log(n))$. To do so, 
we must first provide some definitions and fix some parameters. 
Throughout this subsection, for an arbitrary $\gamma\in (0,1)$ we set $\lambda_0=\gamma$, $\rho=\gamma^3$, $q=\frac{\gamma^2}{2}$, and $\delta=\frac{\gamma^6}{4}$. Notice that we have $\frac{\gamma\lambda_0}{2\rho} = \frac{\gamma^2}{2\gamma^3}=\frac{1}{2\gamma}\geq\frac{\gamma^2}{2}=q$. 
%and hence Corollary \ref{cr:iid:main} holds. 
%Also, As proposed by Corollary \ref{cr:iid:main} 
We will then set $\eps = \frac{  \gamma^2 q^2 \rho \lambda_0 }{2\log \frac {2} {\delta}} = \frac{\gamma^{10}}{8\log(\frac 8 {\gamma^6})}= \frac{\gamma^{10}}{24\log(\frac 2 {\gamma^2})}$; this will be the value $\eps$ we require in the no-superstars assumption.  We note that we have not aimed to optimize these parameters.  
%$ \in O\big(\gamma^{-10}\log \frac 1 {\gamma}\big)$

Set $c = \frac{1-\lambda_0}{\rho}$.  We let $t_0,\dots, t_c$ be the (unique) sequence of thresholds such that, for each $\i \in \{0,\dots,c\}$, we have $\prob{\maxall \leq t_{\i}} = \lambda_0+\i \rho$.  That is, the probability that $\maxall$ falls between any two consecutive thresholds is $\rho$, and the probability that it falls below $t_0$ is $\lambda_0$.
%where $c=\frac{1-\lambda_0}{\rho}$. %Note that we have $\prob{\maxall \leq t_{\i}} \in [0,1-(\lambda_0+\rho)]$. %Note that any pair $t_{\i}$ and $t_{\i+1}$ corresponds to a $\lambda$
 
The next definition captures our desire to combine multiple distributions $D_i$ into a single collection, and study the maximum of the values drawn from that collection of distributions.
 
\begin{definition}[Collection Distribution] Let $S \subseteq \{D_1,\ldots,D_n\}$ be an arbitrary set. We define the collection distribution $D_S$ using the following procedure: $D_S$ draws $x_i$ from distribution $D_i$ for each $D_i \in S$, then returns $\max_{D_i\in S} x_i$.  We use $x_S$ to indicate an outcome of $D_S$. %, where $x_i$ is a draw from $D_i$. 
\end{definition}

The following lemma provides a concentration result for the distribution $D_S$, when $S$ is a set of size $qn$ chosen uniformly at random without replacement from $D_1,\dots,D_n$. Intuitively, this says that if we decompose a random order sequence of $D_1,\dots,D_n$ into $\frac 1 q$ subsequences, each of size $qn$, these subsequences behave similarly to an i.i.d.\ distribution. We use this to prove our main result.  We defer the proof to Section~\ref{subsec:reduc-conc}.

\begin{lemma}\label{lm:Ralg:DS}
	Let $S$ be a set of size $qn$, chosen uniformly at random without replacement from $D_1,\dots,D_n$. 
	With probability $1-\frac{\gamma^3} 2$ for all $\i \in \{0,\dots,c-1\}$ we have
	\begin{align*}
		(1-3\gamma)q \sum_{i=1}^n p^{\i}_i \leq \prob{ t_{\i} \leq x_S < t_{\i+1}} \leq (1+\gamma)q \sum_{i=1}^n p^{\i}_i,
	\end{align*}
	where $p^{\i}_i = \prob{t_{\i} \leq x_i < t_{\i+1}}$, assuming the no $\eps$-superstars assumption with $\eps =\frac{\gamma^{10}}{24\log(\frac 2 {\gamma^2})}$.
\end{lemma}

We use the following definitions in our proof of Theorem \ref{thm:Ralg:main}.

\begin{definition}
	For a given number $x \geq t_0$, we write $\tilde{x} = \max\{ t_{\i} \colon t_{\i} \leq x \}$.  That is, $\tilde{x}$ is $x$ rounded down to the nearest $t_{\i}$.
%	Given a sequence $t_0,\dots,t_c$ and a number, $\tilde{x}$ is the largest $t_{\i}$ smaller than $x$. In other words, it rounds down $x$ to a $t_{\i}$.
Similarly, for a distribution $D$ we use $\tilde{D}$ to represent the distribution that draws $x$ from $D$ and then returns $\tilde{x}$. %round down the outcome of $D$ to a $t_{\i}$.  
\end{definition}

\begin{definition}
	We define a distribution $D_{\min}$ as follows: for any $\i \in \{0,\dots,c-1\}$, $D_{\min}$ returns $t_{\i}$ with probability $(1-3\gamma)q \sum_{i=1}^n p^{\i}_i$, and otherwise $D_{\min}$ returns $0$. 
\end{definition}

\begin{definition}\label{def13}
For $\j \in \{1,\dotsc,\frac 1 q \}$, let $S_{\j}$ be the set of distributions $D_{\pi_{(\j-1) qn+1}}, \dots D_{\pi_{(\j) qn}}$. Let $\hat D_{S_{\j}}$ be a distribution that returns $\tilde x_{S_{\j}}$ with probability $1-4\gamma$ and returns $0$ otherwise. We use $\hat x_{S_{\j}}$ to indicate an outcome of $\hat D_{S_{\j}}$.
\end{definition}

%%%%%%%%%%%%%%%%%%%%%%%%%%%%%%%%%%%%%%%%%%%%%%%%%%%%%%%%%%%%%%%%%%%%%%%%%%

We now present results for two algorithms, Algorithm~\ref{Alg0} and Algorithm~\ref{Alg}, whose pseudocode is listed in the text.
These algorithms take, as parameters, a sequence of thresholds defining an arbitrary threshold-based algorithm for the i.i.d.\ setting
with $1/q$ observations.
Algorithm~\ref{Alg0} provides an intermediary result.  In particular, 
Algorithm~\ref{Alg0} is meant to work with the values $\tilde x_{S_{\j}}$, which recall are ``rounded down'' 
values drawn from the collection distributions.  This algorithm is used to bound the success rate 
if we used the $1/q$ collection distributions to generate our input instead of the actual observations.
We then show that Algorithm~\ref{Alg}, which works with the real observations, performs 
nearly as well as Algorithm~\ref{Alg0}.  

\SetKwInOut{Parameters}{Parameters}
\SetKwInOut{Given}{Given}

\begin{algorithm}[t]
	\Parameters{Thresholds $\tau_1, \dotsc, \tau_{1/q}$}
	\KwIn{Iteratively receive values $\tilde x_{S_{\j}}$, for $\j \in \{1, \dotsc, \frac 1 q\}$.}
	%\KwOut{.}
	%\hrulefill
	
	\begin{algorithmic}[1]
		\STATE With probability $4\gamma$, do not pick $\tilde x_{S_{\j}}$ and move to the next number.\label{Alg0:Ralg:l1}
		\STATE Set $t_0$ such that $\prob{\maxall \leq t_{0}} = \lambda_0$.
		\IF{$\tilde x_{S_{\j}}\leq t_0$}\label{Alg0:Ralg:if1}
		\STATE Do not pick $\tilde x_{S_{\j}}$ and move to the next number.
		\ENDIF
		\IF{$\tilde x_{S_{\j}}\leq \tau_{\j}$}
		\STATE Do not pick $\tilde x_{S_{\j}}$ and move to the next number.
		\ELSE
		\STATE Pick $\tilde x_{S_{\j}}$.\label{Alg:Ralg:end}
		\ENDIF
		
	\end{algorithmic}
	\caption{}
	\label{Alg0}
\end{algorithm}

\begin{algorithm}[t]
	\Parameters{Thresholds $\tau_1, \dotsc, \tau_{1/q}$}
	\KwIn{Iteratively receive values $x_{\pi_i}$, for $i \in \{1, \dotsc, n\}$.}
	%\KwOut{.}
	%\hrulefill
	
	\begin{algorithmic}[1]
		\STATE With probability $4\gamma$, do not pick $x_{\pi_i}$ and move to the next number.
		\STATE Set $t_0$ such that $\prob{\maxall \leq t_{0}} = \lambda_0$.
		\IF{$\tilde x_{\pi_i}\leq t_0$}
		\STATE Do not pick $x_{\pi_i}$ and move to the next number.
		\ENDIF
		\IF{$\tilde x_{\pi_i}\leq \tau_{\lceil qi \rceil}$}
		\STATE Do not pick $x_{\pi_i}$ and move to the next number.
		\ELSE
		\STATE Pick $x_{\pi_i}$.
		\ENDIF
		
	\end{algorithmic}
	\caption{}
	\label{Alg}
\end{algorithm}

% The following lemma shows the effectiveness of Algorithm~\ref{Alg0}. We use this in Theorem \ref{thm:Ralg:main} to prove the efficiency of Algorithm~\ref{Alg}. 

We first show that Algorithm~\ref{Alg0} can simulate an arbitrary i.i.d.\ algorithm with minimal loss, under a no-superstars assumption.

\begin{lemma}\label{lm:Ralg:Alg0}
 	Let $\alg_{\tau}$ be any threshold-based algorithm that selects the maximum with probability at least $\alpha$ for $1/q$ instances of $D_{\min}$, with thresholds $\tau_{1},\dots,\tau_{1/q}$. For any arbitrary $\gamma\in (0,1)$, Algorithm \ref{Alg0} selects the maximum with probability at least $(\alpha-10\gamma)$ for $\tilde D_{S_1},\dots,\tilde D_{S_{1/q}}$, assuming the no $\eps$-superstars assumption with $\eps =\frac{\gamma^{10}}{24\log(\frac 2 {\gamma^2})}$.
\end{lemma}
\begin{proof} 
	First of all notice that the probability that the maximum is less than $t_0$ is $\lambda_0=\gamma$. We assume that any number less than $t_0$ is $0$ and we do not pick it. We miss the maximum with probability $\gamma$ due to this assumption. Algorithm \ref{Alg0} handles this assumption by the condition in line \ref{Alg0:Ralg:if1}.

	%Note that Algorithm \ref{Alg0} do not pick the maximum due to line \ref{Alg0:Ralg:l1} is at most $4\gamma$. In fact, run of Algorithm \ref{Alg0} on $\frac 1 q$ instances of $\tilde D_{S_{\j}}$ is equivalent to run of Lines \ref{Alg0:Ralg:if1} to \ref{Alg:Ralg:end} on $\frac 1 q$ instances of $\hat D_{S_{\j}}$.

By Lemma \ref{lm:Ralg:DS} with probability $1-\frac{\gamma^3} 2$ for all $\i \in \{0,\dots,c\}$ we have
\begin{align}\label{eq:Ralg:x<x<x}
	\prob{\hat x_{S_{\j}}=t_{\i}} \leq \prob{x_{\min}=t_{\i}} \leq \prob{\tilde x_{S_{\j}}=t_{\i}},
\end{align}
where the first inequality follows from $(1-4\gamma)(1+\gamma) \leq 1-3\gamma$ (where $1+\gamma$ and $1-3\gamma$ are coming from Lemma 10 and $1-4\gamma$ is coming from the definition of $\hat{x}_{S_{\j}}$ i.e. Definition 13).
%{\bf MM:  Please explain where the $(1-4\gamma)$ and $(1+\gamma)$ come from.}
%
By the union bound this holds for all $\j \in \{1,\dots,\frac 1 q \}$ and all $\i \in \{0,\dots,c\}$ with probability at least $1-\frac 1 q \frac{\gamma^3} 2 = 1-\gamma$. In the rest of the proof we assume that Inequality \ref{eq:Ralg:x<x<x} holds for all $\j \in \{1,\dotsc,\frac 1 q \}$ and all $\i \in \{0,\dots,c\}$.

%\textbf{Define probability of success and $\sigma$ and $\phi$}
We define $\phi_{\j}$ to be the probability that $\alg_{\tau}$ reaches the $\j$-th number when running on $\frac 1 q$ instances of $D_{\min}$. Similarly, we define $\tilde \phi_{\j}$ to be the probability Algorithm \ref{Alg0} reaches the $\j$-th number when running on $\tilde D_{S_1},\dots,\tilde D_{S_{1/q}}$. 
We also define $\sigma_{\j}$ to be the probability that Algorithm $\alg_{\tau}$, conditioned on reaching the $\j$-th number, accepts the $\j$-th number when running on $\frac 1 q$ instances of $D_{\min}$ and succeeds. Similarly, we define $\tilde \sigma_{\j}$ to be the probability Algorithm \ref{Alg0}, conditioned on reaching the $\j$-th number, accepts the $\j$-th number when running on $\tilde D_{S_1},\dots,\tilde D_{S_{1/q}}$ and succeeds. We refer to this notion as the probability of success at $\j$. Notice that the probability that $\alg_{\tau}$ and Algorithm \ref{Alg0} succeed are $\sum_{\j=1}^{1/q} \phi_{\j}\sigma_{\j}$ and $\sum_{\j=1}^{1/q} \tilde\phi_{\j}\tilde\sigma_{\j}$ respectively. 

In fact, running Algorithm \ref{Alg0} on $\tilde D_{S_{1}},\dots,\tilde D_{S_{1/q}}$ is equivalent to running Lines \ref{Alg0:Ralg:if1} to \ref{Alg:Ralg:end} on $\hat D_{S_{1}},\dots,\hat D_{S_{1/q}}$. %Moreover, Condition \ref{Alg0:Ralg:if1} just increases the probability of reaching the $\j$-th number. 
Hence by inequality \ref{eq:Ralg:x<x<x} we have
\begin{align*}
	\sum_{\j=1}^{1/q} \tilde\phi_{\j} \tilde\sigma_{\j} \geq 
	\sum_{\j=1}^{1/q} \phi_{\j} \tilde\sigma_{\j}.
\end{align*}

Now, let $\j \in \{1,\dots,\frac 1 q\}$ be an arbitrary index. Assume for all $\j'  \in \{1,\dots,\frac 1 q\} \setminus \{\j\}$ we replace distributions $\tilde D_{S_{\j'}}$ with $D_{\min}$. By Inequality \ref{eq:Ralg:x<x<x} this increases the the probability of success at $\j$ by at most a factor $\frac{1}{1-4\gamma}$. Next, if we replace $\tilde D_{S_{\j}}$ with $D_{\min}$ the probability of success at $\j$ decreases and becomes $(1-4\gamma)\sigma_{\j}$. Thus, we have $\frac{1}{1-4\gamma} \tilde\sigma_{\j} \geq (1-4\gamma)\sigma_{\j}$, which implies $\tilde\sigma_{\j} \geq (1-4\gamma)^2\sigma_{\j} \geq (1-8\gamma)\sigma_{\j}$. Therefore we have 
\begin{align*}
		\sum_{\j=1}^{1/q} \phi_{\j} \tilde\sigma_{\j} \geq (1-8\gamma)	\sum_{\j=1}^{1/q} \phi_{\j} \sigma_{\j}&\geq (1-8\gamma) \alpha \geq \alpha-8\gamma.
\end{align*}
Remember that as we mentioned in the beginning, Algorithm \ref{Alg0} misses the maximum with probability $\gamma$ due to the condition in line \ref{Alg0:Ralg:if1}, and it loses another $\gamma$ probability by assuming that Inequality \ref{eq:Ralg:x<x<x} holds for all $\j \in \{1,\frac 1 q \}$ and all $\i \in \{0,\dots,c\}$.  Hence the probability of selecting the maximum drops to $\alpha - 10\gamma$.
\end{proof}

%%%%%%%%%%%%%%%%%%%%%%%%%%%%%%%%%%%%%%%%%%%%%%%%%%%%%%%%%%%%%%%%%%%%%%%%%%

We now want to prove that Algorithm~\ref{Alg} can likewise simulate an arbitrary i.i.d.\ algorithm with minimal loss, by comparing to the performance of Algorithm~\ref{Alg0}.  Recall that Algorithm~\ref{Alg} attempts to simulate Algorithm~\ref{Alg0} by applying threshold $\tau_{\j}$ to each of the $qn$ values in collection $\j$.  There are two ways that this simulation might fail.  First, it might be that two values in collection $\j$ are above threshold $\tau_{\j}$, and Algorithm~\ref{Alg} chooses the smaller one.  Second, it could be that the maximum value from two different collections both round to the same value $\tilde{x}$, and Algorithm~\ref{Alg0} chooses the smaller one; this is fine for Algorithm~\ref{Alg0}, since it cares only about the rounded values, but leads to failure for Algorithm~\ref{Alg}.

%We show how to handle these two modes of failure by way of concentration results.  It turns out that it is unlikely that two or more values in any given collection lie above the corresponding threshold.  Furthermore, it is also unlikely that the maximum value in two different collections round to the same $t_{\i}$.  We defer the formal statements and proofs to Section~\ref{subsec:reduc-conc}.  These results in hand, we are now ready to bound the success probability of Algorithm~\ref{Alg}.  The following is a restatement of our main result for the \our prophet secretary problem under a no-superstars assumption, Theorem~\ref{thm:Ralg:main:foreshadow}.

The following two concentration results handle these two modes of failure.  Lemma~\ref{lm:iid:column} shows that it is unlikely that two or more values in any given collection lie above the corresponding threshold.  Lemma~\ref{lm:iid:row} shows that it is unlikely that the maximum value in two different collections round to the same $t_{\i}$.  We defer the proofs to Section~\ref{subsec:reduc-conc}.

\begin{lemma}\label{lm:iid:column}
	Consider arbitrary numbers $\lambda_0,\gamma,\delta,q \in (0,1)$, $\rho \in (0,1-\lambda_0)$. 
	Set $\eps = \frac{  \gamma^2 q^2 \rho \lambda_0 }{2\log \frac {2} {\delta}}$.
	Let $S$ be a set of size $qn$, chosen uniformly at random without replacement from $D_1,\dots,D_n$. 
	Let $\tau^0$ be such that $\prob{\maxall \leq \tau^0}=1-\rho$. 
	Let $y_i$ be a random binary variable that is $1$ if $\tau^0\leq x_i$ and $0$ otherwise. Let $p'_i= \prob{y_i = 1}$. Assuming the no $\eps$-superstars assumption, with probability $1-\delta$ we have 
	\begin{align*}
	\prob{\exists_{i\in S}  y_i = 1}&\leq \frac{2 q}{\lambda_0}\\
	&\text{and}\\
	\prob{\sum_{i\in S} y_i \geq 2}  &\leq \frac{4 q^2}{\lambda_0^2}.
	\end{align*}
\end{lemma}

\begin{lemma}\label{lm:iid:row}
	Consider arbitrary numbers $\rho,\lambda_0 \in (0,1)$ and $\lambda \in [0, 1-(\lambda_0+\rho)]$.
	Let $\tau^0$ and $\tau^1$ be such that $\prob{\maxall \leq \tau^0}=1-(\lambda+\rho)$ and $\prob{\maxall \leq \tau^1}=1-\lambda$. 
	Let $y_i$ be a random binary variable that is $1$ if $\tau^0\leq x_i \leq \tau^1$ and $0$ otherwise. 
	We have
	\begin{align*}
	\prob{\sum_{i=1}^n y_i \geq 2} \leq \frac{\rho^2}{\lambda_0^2}. 	
	\end{align*}
\end{lemma}

These lemmas in hand, we are now ready to bound the success probability of Algorithm~\ref{Alg}.  This is Theorem~\ref{thm:Ralg:main}, which was a restatement of our main result for the \our prophet secretary problem under a no-superstars assumption, Theorem~\ref{thm:Ralg:main:foreshadow}.

\begin{theorem}\label{thm:Ralg:main}
	Let $\alg_{\tau}$ be a threshold based algorithm that selects the maximum with probability at least $\alpha$ for $1/q$ instances of $D_{\min}$, with thresholds $\tau_{1},\dots,\tau_{1/q}$. For any arbitrary $\gamma\in (0,1)$, Algorithm \ref{Alg} selects the maximum with probability at least $(\alpha-13\gamma)$ for $ D_{\pi_1},\dots, D_{\pi_n}$, assuming the no $\eps$-superstars assumption with $\eps =\frac{\gamma^{10}}{24\log(\frac 2 {\gamma^2})}$.
\end{theorem}
\begin{proof}%{Theorem~\ref{thm:Ralg:main}}
	There are two basic differences between Algorithm \ref{Alg0} and Algorithm \ref{Alg}. First, for each of the sets of $qn$ consecutive numbers $S_{\j}$, Algorithm \ref{Alg0} has the privilege to observe the maximum number in the set at once, while Algorithm \ref{Alg} sees the numbers in the set one by one. Second, the input numbers in Algorithm \ref{Alg0} are all rounded to $t_{\i}$'s, but this is not true for the input of Algorithm \ref{Alg}.
	Therefore, there are two cases where Algorithm \ref{Alg0} selects the maximum of the $\tilde x_{S_{\j}}$ but Algorithm \ref{Alg} does not choose the maximum of the $x_{\pi_i}$. 
	\begin{itemize}
		\item Algorithm \ref{Alg0} picks $\tilde x_{S_{\j}}$. There are two numbers $\tau_{\j} < x_i< x_{i'}$ with $i,i' \in S_{\j}$, and Algorithm \ref{Alg} picks $x_i$. 
		\item Algorithm \ref{Alg0} picks $\tilde x_{S_{\j}}$. But there is another $\j'$ such that $\tilde x_{S_{\j}} = \tilde x_{S_{\j'}}=t_{\i}$ but $ x_{S_{\j}} <  x_{S_{\j'}}$. 	
	\end{itemize}
	We show that first case happens with probability at most $2\gamma$ and the second case happens with probability at most $\gamma$. This together with Lemma \ref{lm:Ralg:Alg0} proves the theorem. 
	Notice that the probability of the first case is at most
	\begin{align*}
	\prob{\exists_{i' \in S_{\j}\setminus \{i\}} x_{i'} \geq \max(\tau_{\j},t_0)} 
	&\leq \prob{\exists_{i' \in S_{\j}} x_{i'} \geq \max(\tau_{\j},t_0)} \\
	&\leq \frac{2 q}{\lambda_0} = \gamma, &\text{By Lemma \ref{lm:iid:column}}
	\end{align*}
	where Lemma \ref{lm:iid:column} holds with probability $1-\delta \geq 1 - \gamma$. Hence the first case happens with probability at most $\gamma+\gamma=2\gamma$.
	
	Notice that in the second case for some $\i$ there are at least two numbers $x_i$ (corresponds to $\j$) and $x_{i'}$ (corresponds to $\j'$) such that $t_{\i}\leq x_i \leq x_{i'} \leq t_{\i+1}$. By Lemma \ref{lm:iid:row}, for a particular $\i$ this happens with probability at most $\frac{\rho^2}{\lambda_0^2}$. By the union bound over all choices of $\i$, the second case happens with probability at most $c\frac{\rho^2}{\lambda_0^2} \leq \frac{\rho}{\lambda_0^2} = \frac{\gamma^3}{\gamma^2} = \gamma$.
\end{proof}

Now we are ready to prove Theorem \ref{thm:ps:algEps}, which is an unconditional improvement that holds even without the no-superstars assumption.

\begin{proofof}{Theorem \ref{thm:ps:algEps}}
%\textbf{[BL: I rephrased parts of the proof assuming that we know whether a given value is $D_1$.  However, I haven't added an explicit comment on this modeling assumption, since I'm not completely sure it's necessary.]}
By Theorem~\ref{thm:Ralg:main:foreshadow}, there is a positive constant $\eps > 0$ such that the statement of Theorem~\ref{thm:ps:algEps} holds whenever the distributions satisfy the no $\eps$-superstars assumption.  We will therefore assume that there exists a distribution in the input that violates the no $\eps$-superstars assumption for this positive constant $\eps$.  That is,
		$\prob{i = \argmax_{j=1}^n x_j}\geq \eps$ for some $i$.
%	Note that for some small positive constant $\eps$ Theorem \ref{thm:Ralg:main:foreshadow} shows that the statement of this theorem holds whenever the distributions satisfy the $\eps$-small bid assumption. Hence, in the rest to prove the theorem we assume that there exist a distribution in the input that violates the $\eps$-small bid assumption for some positive constant $\eps$, i.e. 
%	\begin{align*}
%		\exists_{i\in \{1,\dots,n\}} \hspace{0.7cm} \prob{i = \argmax_{j=1}^n x_j}\geq \eps.
%	\end{align*}
	Without loss of generality we assume that this distribution is $D_1$. Let $\tau$ be the threshold selected by the algorithm in Theorem \ref{thm:ps:alg1}. Recall that Theorem \ref{thm:ps:alg1} shows that, for any arbitrary $\eps' > 0 $, there exists a single threshold algorithm that chooses the maximum value with probability at least $\max_\lambda \sum_{k=1}^{\infty}\Big( \frac 1 k \frac{\lambda^k e^{-\lambda}}{k!} \Big) - \eps' $, for the \our prophet secretary problem. For the purpose of this theorem, we set $\eps'= \frac{e^{-1.5}\eps^2}{32}$. 
	% and let $\eps'<< \eps^2$. Next, we prove the theorem in two cases. 
%\textbf{HE: I reminded what $\eps'$ is, explicitly set it, and used it at the end of the proof. }
%{\bf MM:  Brendan commented out the $\eps'$ here, which I think is wrong, but we certainly need to remind the reader what $\eps'$ is, and I don't see it coming in later in the proof, which says to me that something is missing.}
%\textbf{[BL: Whoops!  Yes, I agree.  We need to keep the note on $\eps'$, we should remind what $\eps'$ is, and we should certainly use it somewhere in the proof (with a comment about how we're using the assumption $\eps' << \eps^2$).]}
	We will consider two cases.  In the first case we have $\prob{x_1 < \tau  \text{ and } 1 = \argmax_{j=1}^n x_j}\geq \frac {\eps}{2}$.  In the second case we have $\prob{x_1 \geq \tau}\geq \frac {\eps}{2}$.  Note that we must be in one of these cases, since 
	\[ \prob{x_1 < \tau  \text{ and } 1 = \argmax_{j=1}^n x_j} + \prob{x_1 \geq \tau} \geq \prob{1 = \argmax_{j=1}^n x_j} \geq \eps. \]
	
	\textbf{Case 1. } In this case we apply the single threshold algorithm of Theorem \ref{thm:ps:alg1}, with a slight modification: if $D_1$ is one of the last $\frac{\eps n}{2}$ items, and we reach it, we stop and accept it regardless of its value. Note that the probability that $D_1$ appears in one of the last $\frac{\eps n}{2}$ positions, and at the same time the maximum appears after $D_1$ (and hence also somewhere in the last $\frac{\eps n}{2}$ positions), is at most $\frac{\eps }{2} \times \frac{\eps }{2} \times \frac 1 2 = \frac{\eps^2 }{8}$.  This is an upper bound on the loss of using this modification of the algorithm.  On the other hand, the probability that $D_1$ appears as one of the last $\frac{\eps n}{2}$ items, is the maximum item, and is below the threshold $\tau$ (which also means no item is above the threshold) is at least $\prob{x_1 < \tau  \text{ and } 1 = \argmax_{j=1}^n x_j} \times \frac{\eps }{2} \geq \frac {\eps^2}{4}$.  This is a lower bound on the expected gain of using this modification to the algorithm.  Therefore in this case we improve Theorem \ref{thm:ps:alg1} by at least $\frac{\eps^2}{4} - \frac{\eps^2}{8} = \frac {\eps^2}{8}$.

	\textbf{Case 2. } In this case we show that the analysis of Theorem \ref{thm:ps:alg1} in not tight and hence we provide a better bound for the algorithm with threshold $\tau$. To prove this, we show a constant gap in Inequality \ref{eq:ps:bound}, which directly translates to a constant improvement on the probability of success of the algorithm. Specifically, we consider the case where $D_1$ is the only item above the threshold, but more than one of its corresponding dummy distribution is above the threshold (i.e., $\mathcal{K}'\geq 2$). In this situation, the algorithm certainly selects the maximum; however, in the analysis, we assumed that of the $\mathcal{K}'$ values above the threshold from the dummy distributions, the algorithm would only choose the maximum with probability $\frac 1 {\mathcal{K}'} \leq \frac{1}{2}$ due to the ordering of items.
	Recall that $\prob{\maxall \leq \tau} = e^{-\lambda} > e^{-1.5}$ and hence $\prob{\max_{i=2}^{n} x_i \leq \tau} > e^{-1.5}$. Moreover, note that 	$\frac{\prob{x_1 \geq \tau}}{2}$ is a lower bound on the probability that we see at least one item above the threshold in half of the dummy distribution corresponding to $D_1$ and hence with probability at least $\Big(\frac{\prob{x_1 \geq \tau}}{2}\Big)^2$ we see at least one item above the threshold in the first half of the distributions and at least one in the second half. Thus, we have 
	\begin{align*}
		\prob{\mathcal{K}'\geq 2 \text{ and } x_1 \geq \tau \text{ and }\forall_{i\in \{2,\dots,n\}} x_i < \tau} \geq \Big(\frac{\prob{x_1 \geq \tau}}{2}\Big)^2 \times \prob{\forall_{i\in \{2,\dots,n\}} x_i < \tau}
		 \geq \frac{e^{-1.5}\eps^2}{8},
	\end{align*}
%	{\bf HE: This is a simpler bound than what you have in mined. I added some explanations on this bound above the inequality. }
%{\bf MM:  It is not at all clear to me where the $\Big(\frac{\prob{x_1 \geq \tau}}{2}\Big)^2$ comes from, since for the inequality, we replaced $D_1$ with many many copies.  It looks to me like we're missing an $n^2/\epsilon$ factor somewhere.}
	Therefore, in an event that occurs with probability at least $\frac{e^{-1.5}\eps^2}{8}$, we can improve our bound from something at most $\frac{1}{2}$ to $1$.  This leads to a gap of $\frac{e^{-1.5}\eps^2}{16}$ in Inequality \ref{eq:ps:bound}, and hence a corresponding improvement to Theorem~\ref{thm:ps:alg1}.  %This leads to an improvement of $\frac{e^{-1.5}\eps^2}{16}$ gap in Inequality \ref{eq:ps:bound} as desired. 
	
	Thus, in either case, we obtain an improvement of $\epsilon_0 = \frac{\eps^2}{16e^{1.5}}$ to the bound in Theorem~\ref{thm:ps:alg1}, which says we select the maximum value with probability at least $\max_\lambda \sum_{k=1}^{\infty}\Big( \frac 1 k \frac{\lambda^k e^{-\lambda}}{k!} \Big) - \eps' + \frac{\eps^2}{16e^{1.5}} = \max_\lambda \sum_{k=1}^{\infty}\Big( \frac 1 k \frac{\lambda^k e^{-\lambda}}{k!} \Big)  + \frac{\eps^2}{32 e^{1.5}}$.
\end{proofof}

% BL: Subsection containing proofs of the concentration bounds
% BL: Moved to an appendix
%\input{reduction-concentration}

% BL: I put the links to the subsections into the iid_reduction file

\section{Conclusion and Open Problems}
\label{sec:conclusion}

We have provided several new results for the best-choice problem in the 
prophet inequality and prophet secretary settings, where the goal is to
maximize the probability of selecting the largest value from a sequence of values
drawn independently from known distributions.  Many of our proofs involve 
Poissonization-style arguments, where
we approximate the number of values above a threshold with a Poisson random
variable.  This approach was particularly useful for generalizing results
from the i.i.d.\ setting to the different setting of arbitrary
distributions in a random order.  We believe this approach may be
useful for other related problems.

Our main open problems relate to our most technical result, namely
that, under the no superstars assumption, we can use an algorithm with
multiple thresholds to select the maximum with probability
approximately $0.5801$ in the setting with arbitrary distributions in
a random order.  It is open to determine what probability can be
achieved for arbitrary distributions in a random order without the no
superstars assumption.  A related open question would be to simplify
our proof; it would be interesting to know if there is a more
straightforward argument, and such an argument might more readily lead
to results without the no superstars assumption.  Indeed, we
conjecture the following: that for any $n$, the worst-case
instance of the best-choice prophet secretary problem is an
i.i.d. instance, so in particular the worst-case success probability
matches that of the i.i.d. best-choice problem.

%main (and most difficult) result,
% namely that we can use an algorithm with multiple thresholds to select the maximum
% with probability approximately $0.5801$ in the setting with arbitrary
% distributions in a random order, under the no superstars assumption.
% Our proof required novel concentration bounds for sampling without replacement,
% but it would be interesting to know if there is a more straightforward argument.
% Of course another open problem is to find 

%Our most techically challenging result is that we can use a 
%threshold-based scheme to select the maximum
%with probability approximately $0.5801$ in the setting with arbitrary
%distributions in a random order, in the limit as $n$ grows large, 
%under a technical assumption that no single distribution generates the maximum value with a significant probability.  
%One open problem is to simplify this
%result, and determine whether or not the assumption can be removed.
%We suspect that it can.  In fact, we conjecture something stronger:
%that for any $n$, the worst-case instance of the best-choice prophet 
%secretary problem is an i.i.d. instance, so in particular the worst-case
%success probability matches that of the i.i.d. best-choice problem.

%{\bf MM: Is this another reasonable
%problem?  I don't know....}
Another open problem is to consider
``best-case'' orderings, where the player trying to select the maximum
is allowed to choose the order of the distributions for observation.
Does the ability to choose the ordering provide an advantage over 
random order, in the worst case?  
%Note that if the worst-case instances for
%best-choice prophet secretary are actually i.i.d. then the answer is no.
Even beyond worst-case instances, there is a computational problem of finding the
best ordering.  Can the best ordering for an arbitrary problem instance be found in polynomial time?
%{\bf MM:  Someone else add their own thing here...}

We extended our results to the problem of selecting one of the top $k$
values.  More generally, one could consider the problem of maximizing 
other functions of the rank of the value selected, such as minimizing the
expected rank.  One could also study
variants in which multiple values can be selected, subject to a downward-closed
constraint, and the goal is to maximize a function of the set of ranks of the
selected values.  For example, how should one select values subject to a matroid constraint,
so as to maximize the probability that the largest value is among
the values selected?

\newpage

%%%%%%%%%%%%%%%%%%%%%%%%%%%%%

%\bibliographystyle{abbrv}

\bibliographystyle{ACM-Reference-Format}

\begin{thebibliography}{GHK{\etalchar{+}}14}

\bibitem[AEE{\etalchar{+}}17]{abolhassani2017beating}
Melika Abolhassani, Soheil Ehsani, Hossein Esfandiari, MohammadTaghi
  HajiAghayi, Robert Kleinberg, and Brendan Lucier.
\newblock Beating $1-1/e$ for ordered prophets.
\newblock In {\em Proceedings of the 49th Annual ACM Symposium on Theory
  of Computing}, pages 61--71, 2017.

\bibitem[AHL12]{alaei2012online}
Saeed Alaei, MohammadTaghi Hajiaghayi, and Vahid Liaghat.
\newblock Online prophet-inequality matching with applications to ad
  allocation.
\newblock In {\em Proceedings of the 13th ACM Conference on Electronic
  Commerce}, pages 18--35, 2012.

\bibitem[Ala14]{alaei2014bayesian}
Saeed Alaei.
\newblock Bayesian combinatorial auctions: Expanding single buyer mechanisms to
  many buyers.
\newblock {\em SIAM Journal on Computing}, 43(2):930--972, 2014.

%\bibitem[BHZ10]{bateni2010submodular}
%MohammadHossein Bateni, MohammadTaghi Hajiaghayi, and Morteza Zadimoghaddam.
%\newblock Submodular secretary problem and extensions.
%\newblock In {\em Approximation, Randomization, and Combinatorial Optimization.
%  Algorithms and Techniques}, pages 39--52. Springer, 2010.

\bibitem[BIK07]{babaioff2007matroids}
Moshe Babaioff, Nicole Immorlica, and Robert Kleinberg.
\newblock Matroids, secretary problems, and online mechanisms.
\newblock In {\em Proceedings of the Eighteenth Annual ACM-SIAM symposium on
  Discrete Algorithms}, pages 434--443, 2007.

\bibitem[BM{\etalchar{+}}15]{bardanet2015concentration}
R{\'e}mi Bardenet and Odalric-Ambrym Maillard.
\newblock Concentration inequalities for sampling without replacement.
\newblock {\em Bernoulli}, 21(3):1361--1385, 2015.

\bibitem[B78]{bojdecki1978optimal}
Tomasz Bojdecki.  
\newblock On optimal stopping of a sequence of
independent random variables -- probability maximizing approach.
\newblock {\em Stochastic Processes and Their Applications}, 6(2):153-163, 1978.  

\bibitem[CFH{\etalchar{+}}17]{correa2017posted}
Jos{\'e} Correa, Patricio Foncea, Ruben Hoeksma, Tim Oosterwijk, and Tjark
  Vredeveld.
\newblock Posted price mechanisms for a random stream of customers.
\newblock In {\em Proceedings of the 2017 ACM Conference on Economics and
  Computation}, pages 169--186. ACM, 2017.

\bibitem[CHMS10]{chawla2009sequential}
Shuchi Chawla, Jason~D. Hartline, David~L. Malec, and Balasubramanian Sivan.
\newblock Multi-parameter mechanism design and sequential posted pricing.
\newblock In {\em Proceedings of the 42nd Annual {ACM} Symposium on Theory of
  Computing}, pages 311--320, 2010.

\bibitem[DEH{\etalchar{+}}15]{dehghani2015online}
Sina Dehghani, Soheil Ehsani, MohammadTaghi Hajiaghayi, Vahid Liaghat, and
  Saeed Seddighin.
\newblock Online survivable network design and prophets.
\newblock 2015.

\bibitem[DEH{\etalchar{+}}17]{dehghani_et_al:LIPIcs:2017:7480}
Sina Dehghani, Soheil Ehsani, MohammadTaghi Hajiaghayi, Vahid Liaghat, and
  Saeed Seddighin.
\newblock {Stochastic k-Server: How Should Uber Work?}
\newblock In {\em Proceedings of the International Colloquium on Automata, Languages, and
  Programming}, 2017.

%\bibitem[DFKL16]{DFKL-ArXiv16}
%Paul D{\"u}tting, Michal Feldman, Thomas Kesselheim, and Brendan Lucier.
%\newblock Posted prices, smoothness, and combinatorial prophet inequalities.
%\newblock {\em arXiv preprint arXiv:1612.03161}, 2016.

%\bibitem[DK15]{dutting2015polymatroid}
%Paul D{\"u}tting and Robert Kleinberg.
%\newblock Polymatroid prophet inequalities.
%\newblock In {\em Algorithms-ESA 2015}, pages 437--449. Springer, 2015.

\bibitem[Dyn63]{dynkin1963optimum}
Eugene~B Dynkin.
\newblock The optimum choice of the instant for stopping a {M}arkov process.
\newblock In {\em Soviet Math. Dokl}, volume~4, 1963.

\bibitem[EHKS18]{EHKS18}
Soheil Ehsani, MohammadTaghi Hajiaghayi, Thomas Kesselheim, and Sahil Singla.
\newblock Prophet secretary for combinatorial auctions and matroids.
\newblock In {\em Proceedings of the Twenty-Ninth Annual {ACM-SIAM} Symposium
  on Discrete Algorithms}, pages 700--714, 2018.

\bibitem[EHLM17]{esfandiari2015prophet}
Hossein Esfandiari, MohammadTaghi Hajiaghayi, Vahid Liaghat, and Morteza
  Monemizadeh.
\newblock Prophet secretary.
\newblock In {\em To appear in SIAM Journal on Discrete Mathematics}. 2017.

\bibitem[Fer89]{Fer89}
Thomas~S. Ferguson.
\newblock Who solved the secretary problem?
\newblock {\em Statist. Sci.}, 4(3):282--289, 08 1989.

%\bibitem[FGL15]{feldman2015combinatorial}
%Michal Feldman, Nick Gravin, and Brendan Lucier.
%\newblock Combinatorial auctions via posted prices.
%\newblock In {\em Proceedings of the Twenty-Sixth Annual ACM-SIAM Symposium on
%  Discrete Algorithms}, pages 123--135, 2015.

\bibitem[Fre83]{Free83}
P.~R. Freeman.
\newblock The secretary problem and its extensions: A review.
\newblock {\em International Statistical Review}, 51(2):189--206, 1983.

\bibitem[FSZ15]{FSZ-SODA15}
Moran Feldman, Ola Svensson, and Rico Zenklusen.
\newblock {A simple O (log~log~(rank))-competitive algorithm for the matroid
  secretary problem}.
\newblock In {\em Proceedings of the Twenty-Sixth Annual ACM-SIAM Symposium on
  Discrete Algorithms}, pages 1189--1201, 2015.

%\bibitem[FZ15]{FZ-FOCS15}
%Moran Feldman and Rico Zenklusen.
%\newblock The submodular secretary problem goes linear.
%\newblock In {\em Proceedings of the 56th Annual Symposium on Foundations of Computer Science}, pages 486--505, 2015.

\bibitem[GGLS08]{garg2008stochastic}
Naveen Garg, Anupam Gupta, Stefano Leonardi, and Piotr Sankowski.
\newblock Stochastic analyses for online combinatorial optimization problems.
\newblock In {\em Proceedings of the Nineteenth Annual ACM-SIAM Symposium on
  Discrete Algorithms}, pages 942--951, 2008.

\bibitem[GHK{\etalchar{+}}14]{gobel2014online}
Oliver G{\"o}bel, Martin Hoefer, Thomas Kesselheim, Thomas Schleiden, and
  Berthold V{\"o}cking.
\newblock Online independent set beyond the worst-case: Secretaries, prophets,
  and periods.
\newblock In {\em Proceedings of the International Colloquium on Automata, Languages, and
  Programming}, pages 508--519, 2014.

\bibitem[GM66]{GM66}
John~P. Gilbert and Frederick Mosteller.
\newblock Recognizing the maximum of a sequence.
\newblock {\em J. Amer. Statist. Assoc.}, 61:35--73, 1966.

\bibitem[GM08]{GM-SODA08}
Gagan Goel and Aranyak Mehta.
\newblock Online budgeted matching in random input models with applications to
  adwords.
\newblock In {\em Proceedings of the Nineteenth Annual ACM-SIAM Symposium on
  Discrete Algorithms}, pages 982--991, 2008.

%\bibitem[GRST10]{gupta2010constrained}
%Anupam Gupta, Aaron Roth, Grant Schoenebeck, and Kunal Talwar.
%\newblock Constrained non-monotone submodular maximization: Offline and
%  secretary algorithms.
%\newblock In {\em International Workshop on Internet and Network Economics},
%  pages 246--257. Springer, 2010.

\bibitem[GS17]{GS-IPCO17}
Guru~Prashanth Guruganesh and Sahil Singla.
\newblock Online matroid intersection: Beating half for random arrival.
\newblock In {\em Proceedings of the International Conference on Integer Programming and
  Combinatorial Optimization}, pages 241--253, 2017.

\bibitem[GZ66]{GZ66}
S.M. Gusein-Zade.
\newblock The problem of choice and the optimal stopping rule for a sequence of
  independent trials.
\newblock {\em Theory of Probility and its Applications}, 11:472--476, 1966.

\bibitem[HKP04]{hajiaghayi2004adaptive}
Mohammad~Taghi Hajiaghayi, Robert Kleinberg, and David~C Parkes.
\newblock Adaptive limited-supply online auctions.
\newblock In {\em Proceedings of the 5th ACM Conference on Electronic
  commerce}, pages 71--80, 2004.

\bibitem[HKS07]{hajiaghayi2007automated}
Mohammad~Taghi Hajiaghayi, Robert Kleinberg, and Tuomas Sandholm.
\newblock Automated online mechanism design and prophet inequalities.
\newblock In {\em Proceedings of the 22nd National Conference on Artificial Intelligence}, volume~7, pages 58--65, 2007.

\bibitem[Hoe63]{Hoeffding63}
Wassily Hoeffding.
\newblock Probability inequalities for sums of bounded random variables.
\newblock {\em Journal of the American Statistical Association},
  58(301):13--30, 1963.

\bibitem[K{\etalchar{+}}85]{kennedy1985optimal}
DP~Kennedy et~al.
\newblock Optimal stopping of independent random variables and maximizing
  prophets.
\newblock {\em The Annals of Probability}, 13(2):566--571, 1985.

\bibitem[Ken87]{kennedy1987prophet}
DP~Kennedy.
\newblock Prophet-type inequalities for multi-choice optimal stopping.
\newblock {\em Stochastic Processes and their Applications}, 24(1):77--88,
  1987.

\bibitem[Ker86]{kertz1986comparison}
Robert~P Kertz.
\newblock Comparison of optimal value and constrained maxima expectations for
  independent random variables.
\newblock {\em Advances in applied probability}, 18(02):311--340, 1986.

\bibitem[Kle05]{kleinberg2005multiple}
Robert Kleinberg.
\newblock A multiple-choice secretary algorithm with applications to online
  auctions.
\newblock In {\em Proceedings of the Sixteenth Annual ACM-SIAM Symposium on
  Discrete Algorithms}, pages 630--631, 2005.

\bibitem[KMT11]{KMT-STOC11}
Chinmay Karande, Aranyak Mehta, and Pushkar Tripathi.
\newblock Online bipartite matching with unknown distributions.
\newblock In {\em Proceedings of the Forty-Third Annual ACM Symposium on Theory
  of Computing}, pages 587--596, 2011.

%\bibitem[KMZ15]{KMZ-STOC15}
%Nitish Korula, Vahab Mirrokni, and Morteza Zadimoghaddam.
%\newblock Online submodular welfare maximization: Greedy beats 1/2 in random
%  order.
%\newblock In {\em Proceedings of the Forty-Seventh Annual ACM Symposium on
%  Theory of Computing}, pages 889--898, 2015.

\bibitem[KP09]{KorulaPal-ICALP09}
Nitish Korula and Martin P{\'a}l.
\newblock {Algorithms for secretary problems on graphs and hypergraphs}.
\newblock In {\em Proceedings of the International Colloquium on Automata, Languages, and
  Programming}, pages 508--520, 2009.

\bibitem[KRTV13]{kesselheim2013optimal}
Thomas Kesselheim, Klaus Radke, Andreas T{\"o}nnis, and Berthold V{\"o}cking.
\newblock An optimal online algorithm for weighted bipartite matching and
  extensions to combinatorial auctions.
\newblock In {\em European Symposium on Algorithms}, pages 589--600, 2013.

\bibitem[KS77]{krengel1977semiamarts}
Ulrich Krengel and Louis Sucheston.
\newblock Semiamarts and finite values.
\newblock {\em Bull. Am. Math. Soc}, 1977.

\bibitem[KS78]{krengel1978semiamarts}
Ulrich Krengel and Louis Sucheston.
\newblock On semiamarts, amarts, and processes with finite value.
\newblock {\em Advances in Probability}, 4:197--266, 1978.

%\bibitem[KW12]{KW-STOC12}
%Robert Kleinberg and S.~Matthew Weinberg.
%\newblock Matroid prophet inequalities.
%\newblock In {\em Proceedings of the 44th Annual Symposium on Theory of Computing
%  Conference}, pages 123--136, 2012.

\bibitem[Lac14]{Lachish-FOCS14}
Oded Lachish.
\newblock O(log log rank) competitive ratio for the matroid secretary problem.
\newblock In {\em Proceedings of the Fifty-Fifth Annual IEEE Symposium on
  Foundations of Computer Science}, pages 326--335, 2014.

\bibitem[LC60]{le1960approximation}
Lucien Le~Cam.
\newblock An approximation theorem for the Poisson binomial distribution.
\newblock {\em Pacific Journal of Mathematics}, 10(4):1181--1197, 1960.

\bibitem[MY11]{MY-STOC11}
Mohammad Mahdian and Qiqi Yan.
\newblock Online bipartite matching with random arrivals: an approach based on
  strongly factor-revealing lps.
\newblock In {\em Proceedings of the Forty-Third Annual ACM Symposium on Theory
  of Computing}, pages 597--606, 2011.

\bibitem[Mey01]{Meyerson-FOCS01}
Adam Meyerson.
\newblock Online facility location.
\newblock In {\em Proceedings of the 42nd Annual Symposium on Foundations of Computer Science}, pages 426--431, 2001.

\bibitem[MU05]{mitzenmacher2005probability}
Michael Mitzenmacher and Eli Upfal.
\newblock {\em Probability and computing: Randomized algorithms and
  probabilistic analysis}.
\newblock Cambridge University Press, 2005.

\bibitem[Pet81]{Petruccelli81}
Joseph D. Petruccelli.
\newblock Best-Choice Problems Involving Uncertainty of Selection and Recall of Observations
\newblock {\em Journal of Applied Probability}, 18(2):415--425, 1981.

\bibitem[Por87]{Porosinski87}
Zdzis\l{}aw Porosi\'{n}ski.
\newblock The full-information best choice problem with a random number of observations.
\newblock {\em Stochastic Processes Applications}, 24:293--307, 1987.


\bibitem[Sam84]{samuel1984comparison}
Ester Samuel-Cahn et~al.
\newblock Comparison of threshold stop rules and maximum for independent
nonnegative random variables.
\newblock {\em the Annals of Probability}, 12(4):1213--1216, 1984.

%\bibitem[RS17]{RS-SODA17}
%Aviad Rubinstein and Sahil Singla.
%\newblock Combinatorial prophet inequalities.
%\newblock In {\em Proceedings of the Twenty-Eighth Annual ACM-SIAM Symposium on
%  Discrete Algorithms}, pages 1671--1687, 2017.

%\bibitem[Rub16]{rubinstein2016beyond}
%Aviad Rubinstein.
%\newblock Beyond matroids: secretary problem and prophet inequality with
%  general constraints.
%\newblock In {\em Proceedings of the 48th Annual {ACM} Symposium on
%  Theory of Computing}, pages 324--332, 2016.

\bibitem[Ste94]{steele1994cam}
J~Michael Steele.
\newblock Le Cam's inequality and Poisson approximations.
\newblock {\em The American Mathematical Monthly}, 101(1):48--54, 1994.

\bibitem[VDVW96]{van1996weak}
Aad~W Van Der~Vaart and Jon~A Wellner.
\newblock Weak convergence.
\newblock In {\em Weak convergence and empirical processes}, pages 16--28.
Springer, 1996.


\bibitem[Yan11]{yan2011mechanism}
Qiqi Yan.
\newblock Mechanism design via correlation gap.
\newblock In {\em Proceedings of the Twenty-Second Annual ACM-SIAM Symposium on
  Discrete Algorithms}, pages 710--719, 2011.

\end{thebibliography}

\newcommand{\etalchar}[1]{$^{#1}$}

\newpage

\appendix
\section{Further Related Work}
\label{sec:related}

Starting with the work of
%s of
%Krengel-Sucheston~\cite{krengel1978semiamarts,krengel1977semiamarts}
%and 
Dynkin~\cite{dynkin1963optimum}, there has been a long line of
research on variants of the secretary problem.
%We  only mention  the ones most relevant to this paper.
See the survey by Ferguson~\cite{Fer89} for a light-hearted but thorough 
historical treatment, and the review paper by
Freeman~\cite{Free83} for many generalizations.  

There have likewise been many generalizations of the prophet inequality, since
the initial work of Garling, Krengel, and Sucheston~\cite{krengel1978semiamarts,krengel1977semiamarts}.
One of the first generalizations was the \textit{multiple-choice
prophet inequality}
\cite{kennedy1987prophet,kennedy1985optimal,kertz1986comparison} in
which we are allowed to pick $k$ items and the goal is to maximize
their sum. Alaei~\cite{alaei2014bayesian} gives an almost tight
($1-{1}/{\sqrt{k+3}}$)-approximation algorithm for this problem (the
lower bound is due to~\cite{hajiaghayi2007automated}), where the
{\em approximation factor} is the ratio of the expectation of the
algorithm to the expectation of the optimum. Similarly, the \textit{multiple-choice secretary}
problem was first studied by Hajiaghayi et
al.~\cite{hajiaghayi2004adaptive}, and Kleinberg~\cite{kleinberg2005multiple} gives a $(1-O(\sqrt{1/k}))$-approximation algorithm.

Other than Dynkin~\cite{dynkin1963optimum}, generally follow-up work
considers approximation factors instead of maximizing the probability
of obtaining the best.  An interesting exception is
Bojdecki~\cite{bojdecki1978optimal}, who provides a general approach
for determining the optimal stopping time for choosing the maximum of
a sequence of i.i.d. random varaibles (along with approaches for
finding the optimal stopping time for some related problems). This work
does not determine bounds on the probability of choosing the maximum, as we
do here for the problems we consider.

The research investigating the relation between prophet inequalities
and online auctions is initiated
in~\cite{hajiaghayi2007automated,chawla2009sequential}. This lead to
several interesting follow up works for
matroids~\cite{yan2011mechanism} and
matchings~\cite{alaei2012online}. Meanwhile, the connection between
secretary problems and online auctions is first explored in
Hajiaghayi et al.~\cite{hajiaghayi2004adaptive}. Its generalization
to matroids is considered
in~\cite{babaioff2007matroids,Lachish-FOCS14,FSZ-SODA15} and to
matchings
in~\cite{GM-SODA08,KorulaPal-ICALP09,MY-STOC11,KMT-STOC11,kesselheim2013optimal,GS-IPCO17}.

%Secretary problems and prophet inequalities have also been studied
%beyond a matroid/matching. For the intersection of $p$ matroids,
%Kleinberg and Weinberg~\cite{KW-STOC12} give an $O(p)$-approximation
%prophet inequality. D{\"u}tting and
%Kleinberg~\cite{dutting2015polymatroid} extend this result to
%polymatroids. Feldman et al.~\cite{feldman2015combinatorial} study
%the generalizations to combinatorial auctions. Later, D{\"u}tting et
%al.~\cite{DFKL-ArXiv16} give a general framework to prove such
%prophet inequalities.
% Submodular variants of the secretary problem have been considered in~\cite{bateni2010submodular,gupta2010constrained,FZ-FOCS15,KMZ-STOC15}. Prophet  and secretary problems have also been studied for many classical combinatorial problems (see e.g., \cite{Meyerson-FOCS01,garg2008stochastic,gobel2014online,dehghani2015online,dehghani_et_al:LIPIcs:2017:7480}).
%Rubinstein~\cite{rubinstein2016beyond} and
%Rubinstein-Singla~\cite{RS-SODA17} consider these problems for
%arbitrary downward-closed constraints.

In the prophet secretary model, Esfandiari et
al.~\cite{esfandiari2015prophet} give a $(1-1/e)$-approximation in
the special case of a single item. Going beyond $1-1/e$ has been
challenging. Only recently, Abolhasani et
al.~\cite{abolhassani2017beating} and Correa et
al.~\cite{correa2017posted} improve this  factor for the single item
i.i.d. setting. Very recently, Ehsani et al.\cite{EHKS18} extend
prophet secretary  for combinatorial auctions and matroids as well.

%{\bf MM:  As mentioned in e-mail, someone should please go through this and
%clarify/check language -- e.g., make clear when we're talking about approximations,
%what the approximations are to, etc. MH: clarified the definition of approximation}

\section{Appendix: Omitted Proofs from Section~\ref{sec:TopK}}
\label{app:TopK}

We present the proof of Theorem~\ref{thm:topk}, which states that one can solve the top-$k$-choice prophet inequality problem with a failure rate that is exponentially decreasing in $k$.  We restate the theorem below for completeness.

\begin{theorem}
For any $k \geq 1$, there exists a single-threshold algorithm for the top-$k$-choice prophet inequality problem that succeeds with probability at least $1 - 2e^{-\gamma k}$, where $\gamma = (3-\sqrt{5})/2$. 
\end{theorem}
\begin{proof}
We'll begin by showing a bound with a slightly worse constant in the exponent.  We will then describe a way to optimize the constant at the end of the proof.

For a given constant $t$, let $X(t)$ be the random variable corresponding to the number of items $i$ such that $x_i \geq t$.  Choose $\tau$ so that $\ex{X(\tau)} = k/2$.  

The single threshold algorithm with threshold $\tau$ will succeed unless $X(\tau) = 0$ or $X(\tau) > k$.  We note that $X(\tau)$ is the sum of $n$ Bernoulli random variables, where variable $i$ is $1$ with probability $\Pr[ x_i \geq t ]$.  By the additive form of the Chernoff bound, we have that
\[ \Pr[ X(\tau) = 0 ] = \Pr[ X(\tau) \leq \ex{X(\tau)} - k/2 ] < e^{-KL( 0 || k/2n ) \cdot n } \]
where $KL(p || q)$ denotes the Kullback-Leibler (KL) divergence.  Using the bound $KL(p || q) \geq (p-q)^2 / q$ for $p < q$, we have that
\[ \Pr[ X(\tau) = 0 ] < e^{-KL( 0 || k/2n ) \cdot n } < e^{ n \cdot ( k / 2n )^2 / (k / 2n) } = e^{-k/4}. \]
Similarly, we have
\[ \Pr[ X(\tau) > k ] = \Pr[ X(\tau) > \ex{X(\tau)} - k/2 ] < e^{-KL( k/n || k/2n ) \cdot n } < e^{ n \cdot ( k / 2n )^2 / (k / n) } = e^{-k/2} \]
where the second inequality uses the bound $KL(p || q) \geq (p-q)^2 / p$ for $p > q$.  Taking a union bound over these two events completes the proof.
% = \inf\{ t \colon E[X(t)] < k/2\}$.  Then $E[X(\tau)] \geq k/2$, and $\tau$ is the largest such  
%size of the set $\{ i \colon x_i \geq \tau \}$.  

We note that if we choose a threshold $\tau$ so that $\ex{X(\tau)} = \gamma k$ for $\gamma = (3-\sqrt{5})/2$, we obtain
a slightly better probability of success $1 - 2e^{-\gamma k}$ with the same argument.  We have not sought to optimize the constant further. 
\end{proof}

We next present the proof of Theorem~\ref{thm:topk.lb}, which shows that one cannot improve upon this exponential dependence on $k$, regardless of $n$ and even for i.i.d.\ instances.  We restate the theorem below.

\begin{theorem}
There exists a constant $c$ such that, for any fixed $k \geq 1$, no algorithm for the top-$k$-choice prophet inequality problem with identical distributions %\our prophet secretary problem 
selects the maximum with probability more than $1 - e^{-c \cdot k}$.
\end{theorem}

\begin{proof}
Take $n > k$ sufficiently large.  Our problem instance is i.i.d., with distribution $D$ as follows.  
%Distribution $D_i$ is a discrete distribution over the set $\{0,i\}$, with $\Pr[x_i = i] = k/n$ for all $i$.  
With probability $k/n$, distribution $D$ takes a value drawn uniformly from $[1,2]$; with the remaining probability, the value is $0$. 
%with probability $k/n$
We say that an observation is \emph{successful} if it takes on a non-zero value.  In order to describe our analysis more conveniently, we will think of the random process that generates our sequence of observations in the following alternative---but equivalent---way. 
\begin{itemize}
\item We first draw $n$ values uniformly from $[1,2]$, say $v_1 < v_2 < \dotsc < v_n$.  We think of $v_i$ as the value that $x_i$ will take if $x_i$ is non-zero.  
We write $D_i$ for the distribution that takes on value $v_i$ with probability $k/n$ and $0$ otherwise.  We will think of value $x_i$ as being drawn from distribution $D_i$.
\item We choose a permutation $\pi$ on $\{1,\cdots,n\}$;  $\pi(i)$ is the position in the sequence that distribution $D_i$ appears.
\item We choose a number of successes $Z_1$ for the first $n/2$ observations, and correspondingly a 
number of successes $Z_2$ for the second $n/2$ observations.  Both $Z_1$ and $Z_2$ are binomial random variables
$\mbox{Bin}(n/2,k/n)$ and are chosen accordingly.
\item We choose permutations $\sigma_1$ on $\{1,\cdots,n/2\}$ and $\sigma_2$ on $\{n/2+1,\cdots,n\}$;  $\sigma_1$ gives the order of
the successful observations in the first $n/2$ observations, and similarly for $\sigma_2$, as described below.
\end{itemize}

More formally, we see observations in the order
$x_{\pi(1)},\dotsc,x_{\pi(n)}$.  For each $t \in \{1,\cdots,n/2\}$, $x_{\pi(t)}
= v_{\pi(t)}$ if $\sigma_1(t) \leq Z_1$, and otherwise $x_{\pi(t)} = 0.$
Similarly, for each $t \in \{n/2+1,\cdots,n\}$, $x_{\pi(t)} = v_{\pi(t)}$ if
$\sigma_2(t) \leq Z_2$, and otherwise $x_{\pi(t)} = 0.$ 
This process generates a distribution over value sequences that
is identical to the distribution of value sequences in our i.i.d.\ top-$k$-choice problem.

We now consider the following events.  Event $A$ is that $Z_1 = k$;
that is, the first half has $k$ non-zero values.  Event $B$ is that,
for each $t_1,t_2$ satisfying $t_1 \leq n/2$, $t_2 > n/2$,
$\sigma_1(t_1) \leq k$, and $\sigma_2(t_2) \leq k$, we have that
$\pi(t_1) \leq \pi(t_2)$.  That is, event $B$ is that the first $k$
non-zero values in the first half of the observations (as determined by $\sigma_1$) will
be less than the first $k$ non-zero values in the second half (as determined
by $\sigma_2$).  Note that, from the way we have defined event $B$, it is independent of $Z_1$
and $Z_2$, as it depends only on $\pi,\sigma_1,$ and $\sigma_2$.  Because of this,
events $A$ and $B$ are independent of each other (and independent of the value of $Z_2$).  

We make the following claims.  First, each of the events $A$ and $B$ happen with probability
$e^{-\theta(k)}$.  Second, conditioned on both $A$ and $B$ occurring, any algorithm must
fail with probability at least $e^{-\theta(k)}$.  The result follows immediately from these claims.

For event $A$, $Z_1$ is distributed as $\mbox{Bin}(n/2,k/n)$, and a simple calculation shows
that it equals $k$ with probability at least $e^{-c_1k}$ for a suitable constant $c_1$ and
large enough $k$.  Indeed, the distribution is well approximated by a Poisson distribution, so
the desired probability is approximately $e^{-k/2}(k/2)^k/k!,$ which is $e^{-\theta(k)}$.

For event $B$, since $\pi$ is a random ordering on the elements, the probability the first
$k$ values determined by $\sigma_1$ are all less than the first
$k$ values determined by $\sigma_2$ is just ${2k \choose k} \approx 2^{2k}/{\sqrt{\pi k}}$,
which is $e^{-\theta(k)}$.

Now, for any algorithm, consider any realization of $\{v_1, \dotsc, v_n\}$, $\pi$, $\sigma_1$,
$\sigma_2$, and $Z_1$ for which events $A$ and $B$ both occur.  Note
that specifying $Z_2$ then specifies the entire process.  Let us give
the algorithm the additional power to decide, knowing $\{v_1, \dotsc, v_n\}$, $\pi$,
$\sigma_1$, $\sigma_2$, and $Z_1$ (but not $Z_2$), whether to have selected an element
or not after the first $n/2$ observations.  If the algorithm does not
select an item, it will fail when $Z_2 = 0$, as then the $k$ largest
items have all appeared in the first half.  If the algorithm does 
select an item, it will fail when $Z_2 \geq k$, as then the $k$ largest
items all appear in the second half.  
As $Z_2$ is distributed as $\mbox{Bin}(n/2,k/n)$, each of these possibilities 
for $Z_2$ occurs with probability $e^{-\theta(k)}$.
Thus, if we condition on $A$ and $B$ both occurring, the algorithm fails with probability 
$e^{-\theta(k)}$ whether or not it chooses a value from among the first $n/2$ observations, and the result follows.
\end{proof}

\section{Appendix: Omitted Proofs from Section 5}
\label{subsec:reduc-alg-proofs}
\subsection{Concentration Bounds}
\label{subsec:reduc-conc}

This section is dedicated to the proofs of Lemmas~\ref{lm:Ralg:DS}, \ref{lm:iid:column}, and \ref{lm:iid:row}.
To begin, we require several prelimiary lemmata.  The following lemma, for an arbitrary pair of thresholds $\tau^0 \leq \tau^1$, bounds the probability that at least one of the $x_i$'s is within the range $[\tau^0,\tau^1]$.

\begin{lemma}\label{lm:iid:allyprob}
	Consider arbitrary numbers $\rho \in (0,1)$ and $\lambda \in [0, 1-\rho)$.
	Let $\tau^0$ and $\tau^1$ be such that $\prob{\maxall \leq \tau^0}=1-(\lambda+\rho)$ and $\prob{\maxall \leq \tau^1}=1-\lambda$. 
	Let $y_i$ be a random binary variable that is $1$ if $\tau^0\leq x_i \leq \tau^1$ and $0$ otherwise. 
	We have $$\rho \leq \prob{\exists_{i\in \{1,\dots,n\}} y_i = 1} \leq  \frac{\rho}{1-\lambda}.$$
\end{lemma}
\begin{proof}
	On one hand we have
	\begin{align*}
		\prob{\exists_{i\in \{1,\dots,n\}} y_i = 1} \geq \prob{\tau^0 \leq \maxall \leq \tau^1} = \rho.
	\end{align*}
	On the other hand we have
	\begin{align*}
		\lambda+\rho &= \prob{\maxall > \tau^0} \\
		&=  \prob{\maxall > \tau^1} +  \prob{\maxall \leq \tau^1} \times \prob{\exists_{i\in \{1,\dots,n\}} \tau^0\leq x_i \leq \tau^1 \Big| \forall_{i\in \{1,\dots,n\}} x_i \leq \tau^1}\\
		&\geq  \prob{\maxall > \tau^1} +  \prob{\maxall \leq \tau^1} \times \prob{\exists_{i\in \{1,\dots,n\}} \tau^0\leq x_i \leq \tau^1}\\
		&=  \lambda +  (1-\lambda)  \prob{\exists_{i\in \{1,\dots,n\}} \tau^0\leq x_i \leq \tau^1 }\\
		&= \lambda + (1-\lambda) \prob{\exists_{i\in \{1,\dots,n\}} y_i = 1}. 
	\end{align*}
	This implies
	\begin{align*}%\label{eq:iid:yprob}
		\prob{\exists_{i\in \{1,\dots,n\}} y_i = 1} \leq  \frac{\rho}{1-\lambda}.
	\end{align*}
\end{proof}

%%%%%%%%%%%%%%%%%%%%%%%%%%%%%%%%%%%%%%%%%%%%%%%%%%%%%%%%%%%%%%%%%%%%%%%%%%

For an arbitrary index $i$, the following lemma upper bounds the probability that $x_i$ is within the range $[\tau^0,\tau^1]$. Later, we use this to show a concentration bound in Lemma \ref{lm:iid:sumBound}.

\begin{lemma}\label{lm:iid:oneyprob}
	Consider arbitrary numbers $\rho \in (0,1)$ and $\lambda \in [0, 1-\rho)$.
	Let $\tau^0$ and $\tau^1$ be such that $\prob{\maxall \leq \tau^0}=1-(\lambda+\rho)$ and $\prob{\maxall \leq \tau^1}=1-\lambda$. 
	Let $y_i$ be a random binary variable that is $1$ if $\tau^0\leq x_i \leq \tau^1$ and $0$ otherwise.
	Assuming the no $\eps$-superstars assumption we have $$\prob{y_j=1} \leq \frac{\prob{j = \argmax_{i=1}^n x_i}}{1-(\lambda+\rho)} \leq \frac{\eps}{1-(\lambda+\rho)}.$$
\end{lemma}
\begin{proof}
	For any $j$ we have 
	\begin{align*}
		\prob{j = \argmax_{i=1}^n x_i} &\geq \prob{x_j\geq \tau^0} \prob{\argmax_{i\in \{0,\dots,n\}\setminus j} x_i < \tau^0 }\\
		&\geq \prob{x_j\geq \tau^0} \prob{\argmax_{i\in \{0,\dots,n\}} x_i < \tau^0 }\\
		&= \prob{x_j\geq \tau^0} \big(1-(\lambda+\rho)\big)\\
		&\geq \prob{\tau^1 \geq x_j\geq \tau^0} \big( 1-(\lambda+\rho)\big)\\
		&= \prob{y_j=1} \big( 1-(\lambda+\rho)\big).\\
	\end{align*}
	This together with the no-superstars assumption implies that
	\begin{align*}%\label{eq:iid:p'}
		\prob{y_j=1} \leq \frac{\prob{j = \argmax_{i=1}^n x_i}}{1-(\lambda+\rho)} \leq \frac{\eps}{1-(\lambda+\rho)}.
	\end{align*}
\end{proof}

%%%%%%%%%%%%%%%%%%%%%%%%%%%%%%%%%%%%%%%%%%%%%%%%%%%%%%%%%%%%%%%%%%%%%%%%%%

The following lemma, for an arbitrary set $S$ of indices, compares the expected number of $x_i$'s that are in a range $[\tau^0,\tau^1]$ with the probability of observing at least one $x_i$ in the range $[\tau^0,\tau^1]$. We later use this to exchange $\prob{\exists_{i\in S} x_i \in [\tau^0,\tau^1]}$ and $\sum_{i\in S} \prob{ x_i \in [\tau^0,\tau^1]}$.

\begin{lemma}\label{lm:iid:ES}
	Consider arbitrary numbers $\rho \in (0,1)$ and $\lambda \in [0, 1-\rho)$.
	Let $\tau^0$ and $\tau^1$ be such that $\prob{\maxall \leq \tau^0}=1-(\lambda+\rho)$ and $\prob{\maxall \leq \tau^1}=1-\lambda$. 
	Let $y_i$ be a random binary variable that is $1$ if $\tau^0\leq x_i \leq \tau^1$ and $0$ otherwise. Let $p'_i= \prob{y_i = 1}$.
	For any set $S\subseteq  \{1,\dots,n\}$ we have
	$$\max\big(1-\sum_{i\in S} p'_i, 1-\frac{\rho}{1-\lambda}\big)\sum_{i\in S} p'_i \leq \prob{\exists_{i\in S} y_i = 1}\leq \sum_{i\in S} p'_i.$$
\end{lemma}
\begin{proof}
	We have
	\begin{align*}
	\prob{\exists_{i\in S} y_i = 1} &= 1- \prob{\forall_{i\in S} y_i = 0} \\
	&= 1- \Pi_{i\in S} (1-p'_i) \\
	&\geq 1-\exp\Big(-\sum_{i\in S} p'_i\Big).
	\end{align*}
	This implies that 
	\begin{align*}
	\sum_{i\in S} p'_i &\leq \log\Big( \frac{1}{1- \prob{\exists_{i\in S} y_i = 1}} \Big)\\
	 &\leq \frac{1}{1- \prob{\exists_{i\in S} y_i = 1}} -1 & \log(\xi)\leq \xi-1\\
	 &= \frac {\prob{\exists_{i\in S} y_i = 1}}{1-\prob{\exists_{i\in S} y_i = 1}}\\
	 &\leq \frac {\prob{\exists_{i\in S} y_i = 1}}{1-\prob{\exists_{i\in \{1,\dots,n\}} y_i = 1}}\\
	 &\leq \frac {\prob{\exists_{i\in S} y_i = 1}}{1-\frac{\rho}{1-\lambda}}.&\text{Using Lemma \ref{lm:iid:allyprob}}
	 %\\
	 % &\leq \Big(1+2\frac{(1-\lambda)\rho}{\lambda}\Big)\prob{\exists_{i\in S} y_i = 1} &\text{Assuming $\rho \leq 2 \lambda$, i.e., $\frac{(1-\lambda)\rho}{\lambda}\leq \frac 1 2$}\\
	% &\leq \Big(1+\frac{2\rho}{\lambda}\Big)\prob{\exists_{i\in S} y_i = 1}
	\end{align*}
	This implies 
	$$\big(1-\frac{\rho}{1-\lambda}\big)\sum_{i\in S} p'_i \leq \prob{\exists_{i\in S} y_i=1}.$$
	Similarly, we have 
	\begin{align*}
	\sum_{i\in S} p'_i &\leq \frac {\prob{\exists_{i\in S} y_i = 1}}{1-\prob{\exists_{i\in S} y_i = 1}}\\
	&\leq \frac {\prob{\exists_{i\in S} y_i = 1}}{1-\ex{\sum_{i\in S} Y_i}}\\
	&= \frac {\prob{\exists_{i\in S} y_i = 1}}{1-\sum_{i\in S} p'_i},
	\end{align*}
	which implies
	$$\big(1-\sum_{i\in S} p'_i\big)\sum_{i\in S} p'_i \leq \prob{\exists_{i\in S} y_i=1}.$$
	On the other hand we have
	\begin{align*}
	\prob{\exists_{i\in S} y_i = 1} \leq \ex{\sum_{i\in S} y_i} = \sum_{i\in S} p'_i.
	\end{align*}
\end{proof}

%%%%%%%%%%%%%%%%%%%%%%%%%%%%%%%%%%%%%%%%%%%%%%%%%%%%%%%%%%%%%%%%%%%%%%%%%%

In Lemma \ref{lm:iid:sumBound} below we show the concentration of $\sum_{i \in S} \prob{ x_i \in [\tau^0,\tau^1]}$ for a set $S$ chosen uniformly at random without replacement. 
To prove Lemma \ref{lm:iid:sumBound} we use a variation of Massart's inequality for sampling without replacement \cite{van1996weak}. Then to apply Massart's bound to $\sum_{i \in S} \prob{ x_i \in [\tau^0,\tau^1]}$, we use Lemma \ref{lm:iid:oneyprob} to upper bound $\prob{ x_i \in [\tau^0,\tau^1]}$ and use Lemma \ref{lm:iid:allyprob} to lower bound $\ex{\sum_{i \in S} \prob{ x_i \in [\tau^0,\tau^1]}}$.

% {\bf MM: cite Hoeffding 63. and commented citation below.} {\bf HE: Done. We are directly using \cite{bardenet2015concentration}}
% Bardenet, R., & Maillard, O. A. (2015). Concentration inequalities for sampling without replacement. Bernoulli, 21(3), 1361-1385.

\begin{lemma}[Massart's inequality]\label{lm:iid:hf}
	Let $\Psi_1,\dots,\Psi_n$ be a set of $n$ numbers and let $\psi_1,\dots,\psi_{c}$ be a subset of $\Psi_1,\dots,\Psi_n$ drawn uniformly at random without replacement. We have 
%Let $b = \max_i \Psi_i$ and $a = \min_i \Psi_i$.  
	\begin{align*}
		\prob{\Big|\frac 1 c \sum_{i=1}^c \psi_i -  \bar{\Psi} \Big| \geq \gamma} \leq 2\exp\Big(-\frac{c^2 \gamma^2}{\sum_{i=1}^n (\Psi_i-\bar{\Psi})^2}\Big),
	\end{align*}
	where $\bar{\Psi}=\frac 1 n\sum_{i=1}^n \Psi_i $, and $n$ is assumed to be divisible by $c$.
\end{lemma}

Now we are ready to prove Lemma \ref{lm:iid:sumBound}.

%{\bf MM: This is where I broke things, so maybe go through line by line vs an old version and check for OK-ness.}  

\begin{lemma}\label{lm:iid:sumBound}
	Consider arbitrary numbers $\rho,\gamma,\eps,q \in (0,1)$, $\lambda \in [0, 1-\rho)$.
	Let $S$ be a set of size $qn$, chosen uniformly at random without replacement from $D_1,\dots,D_n$. 
	Let $\tau^0$ and $\tau^1$ be such that $\prob{\maxall \leq \tau^0}=1-(\lambda+\rho)$ and $\prob{\maxall \leq \tau^1}=1-\lambda$. 
	Let $y_i$ be a random binary variable that is $1$ if $\tau^0\leq x_i \leq \tau^1$ and $0$ otherwise. Let $p'_i= \prob{y_i = 1}$. Assuming the no $\eps$-superstars assumption, with probability $1-2\exp \Big(-  \frac{  \gamma^2 q^2 \rho (1-(\lambda+\rho))}{2\eps }\Big)$ we have 
	\begin{align*}
		(1-\gamma)q \sum_{i=1}^n p'_i \leq \sum_{i\in S} p'_i \leq (1+\gamma)q \sum_{i=1}^n p'_i 
	\end{align*}
\end{lemma}
\begin{proof}
	Let $z_i$ be a random variable that is $1$ when $i\in S$ and $0$ otherwise. We have
	\begin{align*}
	\sum_{i=0}^n p'_i 
	&\geq \prob{\exists_{i\in S} y_i = 1} &\text{By Lemma \ref{lm:iid:ES}}\\
	&\geq \rho. &\text{By Lemma \ref{lm:iid:allyprob}}
	\end{align*}
	Moreover, by Lemma \ref{lm:iid:oneyprob} we have $0 \leq p'_i \leq \frac{\eps}{1-(\lambda+\rho)}$. Thus,
	\begin{align*}
	&\prob{\Big| \sum_{i\in S} p'_i - q \sum_{i=1}^n p'_i \Big|\geq  \gamma q\sum_{i=1}^n p'_i} =\\
	&\prob{\Big| \frac 1 {qn} \sum_{i\in S} p'_i - \frac 1 n \sum_{i=1}^n p'_i \Big|\geq  \gamma \frac 1 n \sum_{i=1}^n p'_i} = &\text{Multiply both sides by $\frac 1 {qn}$}\\
	%%%%%%%
	%%%%%%%
%	&\prob{\Big| \sum_{i\in S} p'_i - \ex{\sum_{i\in S} p'_i} \Big|\geq  \gamma\ex{\sum_{i\in S} p'_i}} \leq & \ex{\sum_{i\in S} p'_i} = q \sum_{i=1}^n p'_i \\
	& 2\exp \Big(- \frac{(qn)^2\big( \gamma\frac 1 n  \sum_{i=1}^n p'_i\big)^2}{\sum_{i=1}^n \big(p'_i-\frac 1 n  \sum_{i=1}^n p'_i\big)^2} \Big) = &\text{Massart bound}\\
	& 2\exp \Big(- q^2 \gamma^2\frac{ \big( \sum_{i=1}^n p'_i\big)^2}{\sum_{i=1}^n \big(p'_i-\frac 1 n  \sum_{i=1}^n p'_i\big)^2} \Big) \leq \\
%%%%%%
	& 2\exp \Big(- q^2 \gamma^2\frac{ \big( \sum_{i=1}^n p'_i\big)^2}{\sum_{i=1}^n {p'_i}^2+ \sum_{i=1}^n \big(\frac 1 n  \sum_{i=1}^n p'_i\big)^2} \Big) \leq \\
%%%%%%
	& 2\exp \Big(- q^2 \gamma^2\frac{ \big( \sum_{i=1}^n p'_i\big)^2}{2\sum_{i=1}^n {p'_i}^2} \Big) \leq \\
%%%%%%
	& 2\exp \Big(- q^2 \gamma^2\frac{ \big( \sum_{i=1}^n p'_i\big)^2}{ 2\frac{\eps}{1-(\lambda+\rho)} \sum_{i=1}^n p'_i} \Big) = & p'_i \leq \frac{\eps}{1-(\lambda+\rho)}\\
%%%%%%
	& 2\exp \Big(- \frac{ q^2 \gamma^2(1-(\lambda+\rho))}{2\eps} \sum_{i=1}^n p'_i  \Big)\leq \\
%%%%%
& 2\exp \Big(-  \frac{  \gamma^2 q^2 \rho (1-(\lambda+\rho))}{2\eps }\Big) & \sum_{i=1}^n p'_i \geq {\rho}\\
\end{align*}
	
\end{proof}

%%%%%%%%%%%%%%%%%%%%%%%%%%%%%%%%%%%%%%%%%%%%%%%%%%%%%%%%%%%%%%%%%%%%%%%%%%

Next, we use Lemma \ref{lm:iid:sumBound} together with Lemma \ref{lm:iid:ES} to show the concentration of $\prob{\exists_{i \in S} x_i \in [\tau^0,\tau^1]}$ for a set $S$ chosen uniformly at random without replacement.

\begin{lemma}\label{lm:iid:main}
	Consider arbitrary numbers $\rho,\gamma,\eps,q \in (0,1)$, $\lambda \in [0, 1-\rho)$.
	Let $S$ be a set of size $qn$, chosen uniformly at random without replacement from $D_1,\dots,D_n$. 
	Let $\tau^0$ and $\tau^1$ be such that $\prob{\maxall \leq \tau^0}=1-(\lambda+\rho)$ and $\prob{\maxall \leq \tau^1}=1-\lambda$. 
	Let $y_i$ be a random binary variable that is $1$ if $\tau^0\leq x_i \leq \tau^1$ and $0$ otherwise. Let $p'_i= \prob{y_i = 1}$. Assuming the no $\eps$-superstars assumption, with probability $1-2\exp \Big(-  \frac{  \gamma^2 q^2 \rho (1-(\lambda+\rho))}{2\eps }\Big)$ we have 
	\begin{align*}
	\big(1-\gamma-\frac {2q\rho}{1-(\lambda+\rho)}\big)q \sum_{i=1}^n p'_i \leq \prob{\exists_{i\in S} y_i = 1} \leq (1+\gamma)q \sum_{i=1}^n p'_i 
	\end{align*}
\end{lemma}
\begin{proof}
	With probability $1-2\exp \Big(-  \frac{ \gamma^2 q^2 \rho (1-(\lambda+\rho))}{2\eps }\Big)$ Lemma \ref{lm:iid:sumBound} holds and we have
	\begin{align*}
	\prob{\exists_{i\in S} y_i = 1} &\leq \sum_{i\in S} p'_i &\text{By Lemma \ref{lm:iid:ES}}\\
	&\leq (1+\gamma)q \sum_{i=1}^n p'_i. &\text{By Lemma \ref{lm:iid:sumBound}}
	\end{align*}
	Moreover, we have
	\begin{align*}
	\prob{\exists_{i\in S} y_i = 1} &\geq \big(1-\sum_{i\in S} p'_i\big)\sum_{i\in S} p'_i &\text{By Lemma \ref{lm:iid:ES}}\\
	&\geq \big(1-\sum_{i\in S} p'_i\big)(1-\gamma)q\sum_{i=1}^n p'_i &\text{By Lemma \ref{lm:iid:sumBound}}\\
	&\geq \big(1-(1+\gamma)q\sum_{i=1}^n p'_i\big)(1-\gamma)q\sum_{i=1}^n p'_i &\text{By Lemma \ref{lm:iid:sumBound}}\\
	&\geq \big(1-(1+\gamma)q\frac{1-\lambda}{1-(\lambda+\rho)}\prob{\exists_{i\in S} y_i = 1}\big)(1-\gamma)q\sum_{i=1}^n p'_i &\text{By Lemma \ref{lm:iid:ES}}\\
	&\geq \big(1-2q\frac{1-\lambda}{1-(\lambda+\rho)}\prob{\exists_{i\in S} y_i = 1}\big)(1-\gamma)q\sum_{i=1}^n p'_i\\ 
	&\geq \big(1-2q\frac{1-\lambda}{1-(\lambda+\rho)}\frac{\rho}{1-\lambda}\big)(1-\gamma)q\sum_{i=1}^n p'_i &\text{By Lemma \ref{lm:iid:allyprob}} \\
	&\geq \big(1-\frac{2q\rho}{1-(\lambda+\rho)}\big)(1-\gamma)q\sum_{i=1}^n p'_i\\
	&\geq \big(1-\gamma-\frac{2q\rho}{1-(\lambda+\rho)}\big)q\sum_{i=1}^n p'_i.
	\end{align*}
	
\end{proof}

%%%%%%%%%%%%%%%%%%%%%%%%%%%%%%%%%%%%%%%%%%%%%%%%%%%%%%%%%%%%%%%%%%%%%%%%%%

The following corollary is a simplified (and restricted) variation of Lemma \ref{lm:iid:main}.

\begin{corollary}\label{cr:iid:main}
	Consider arbitrary numbers $\rho,\lambda_0,\gamma,\delta \in (0,1)$, $\lambda \in [0,1-(\lambda_0+\rho)]$ and $q \in \big(0,\min\big(\frac{\gamma\lambda_0}{2\rho},1\big)\big)$. 
	Set $\eps = \frac{  \gamma^2 q^2 \rho \lambda_0 }{2\log \frac {2} {\delta}}$.
	Let $S$ be a set of size $qn$, chosen uniformly at random without replacement from $D_1,\dots,D_n$. 
	Let $\tau^0$ and $\tau^1$ be such that $\prob{\maxall \leq \tau^0}=1-(\lambda+\rho)$ and $\prob{\maxall \leq \tau^1}=1-\lambda$. 
	Let $y_i$ be a random binary variable that is $1$ if $\tau^0\leq x_i \leq \tau^1$ and $0$ otherwise. Let $p'_i= \prob{y_i = 1}$. Assuming the no $\eps$-superstars assumption, with probability $1-\delta$ we have 
	\begin{align*}
	(1-2\gamma)q \sum_{i=1}^n p'_i \leq \prob{\exists_{i\in S} y_i = 1} \leq (1+\gamma)q \sum_{i=1}^n p'_i 
	\end{align*}
\end{corollary}
\begin{proof}
	Note that Lemma \ref{lm:iid:main} holds with probability
	\begin{align}\label{eq:iid:delta}
		1-{2}\exp \Big(-  \frac{  \gamma^2 q^2 \rho (1-(\lambda+\rho))}{2\eps  }\Big) &\geq
		1-{2}\exp \Big(-  \frac{  \gamma^2 q^2 \rho \lambda_0}{2\eps  }\Big)\nonumber\\
		&=
		1-{2}\exp \Bigg(-  \frac{  \gamma^2 q^2 \rho \lambda_0}{2\big( \frac{  \gamma^2 q^2 \rho \lambda_0 }{2\log \frac {2} {\delta}}  \big)}\Bigg)\nonumber\\
		& = 1- {2}\exp\big(-\log \frac {2} {\delta}\big) \nonumber\\
		&= 1-\delta.
	\end{align}
	Note that Lemma \ref{lm:iid:main} directly gives us $\prob{\exists_{i\in S} y_i = 1} \leq (1+\gamma)q \sum_{i=1}^n p'_i $. Moreover, we have
	\begin{align*}
	\prob{\exists_{i\in S} y_i = 1} &\geq \big(1-\gamma-\frac {2q\rho}{1-(\lambda+\rho)}\big)q \sum_{i=1}^n p'_i &\text{Lemma \ref{lm:iid:main}}\\
	&\geq \big(1-\gamma-\frac {2q\rho}{\lambda_0}\big)q \sum_{i=1}^n p'_i \\	
	&\geq \big(1-\gamma-\frac {2\frac{\gamma\lambda_0}{2\rho}\rho}{\lambda_0}\big)q \sum_{i=1}^n p'_i \\
	&= (1-2\gamma)q \sum_{i=1}^n p'_i.	
	\end{align*} 
\end{proof}

%%%%%%%%%%%%%%%%%%%%%%%%%%%%%%%%%%%%%%%%%%%%%%%%%%%%%%%%%%%%%%%%%%%%%%%%%%

We can now prove Lemma~\ref{lm:Ralg:DS}.  We will restate it as Lemma~\ref{lm:Ralg:DS:proof} below for convenience. Recall that for the purpose of this lemma for some arbitrary $\gamma\in (0,1)$ we set $\lambda_0=\gamma$, $\rho=\gamma^3$, $q=\frac{\gamma^2}{2}$, and $\delta=\frac{\gamma^6}{4}$.

\begin{lemma}\label{lm:Ralg:DS:proof}
	Let $S$ be a set of size $qn$, chosen uniformly at random without replacement from $D_1,\dots,D_n$. 
	With probability $1-\frac{\gamma^3} 2$ for all $\i \in \{0,\dots,c-1\}$ we have
	\begin{align*}
		(1-3\gamma)q \sum_{i=1}^n p^{\i}_i \leq \prob{ t_{\i} \leq x_S < t_{\i+1}} \leq (1+\gamma)q \sum_{i=1}^n p^{\i}_i,
	\end{align*}
	where $p^{\i}_i = \prob{t_{\i} \leq x_i < t_{\i+1}}$, assuming the no $\eps$-superstars assumption with $\eps =\frac{\gamma^{10}}{24\log(\frac 2 {\gamma^2})}$.
\end{lemma}
\begin{proof}
	Note that by Corollary \ref{cr:iid:main}, for a fixed $\i \in \{0,\dots,c-1\}$ with probability $1-\delta = 1- \frac{\gamma^6}{4}$ we have
	\begin{align}\label{eq:Ralg:t1t2}
			(1-2\gamma)q \sum_{i=1}^n p^{\i}_i \leq \prob{ \exists_{i\in S} t_{\i} \leq x_i < t_{\i+1}} \leq (1+\gamma)q \sum_{i=1}^n p^{\i}_i.
	\end{align}
	By the union bound, this holds for all $\i \in \{0,\dots,c-1\}$ with probability at least
	\begin{align*}
		1- c \frac{\gamma^6}{4} = 1- \frac{1-\lambda_0}{\rho} \frac{\gamma^6}{4} \geq 1 - \frac{\gamma^6}{4 \rho} = 1- \frac{\gamma^3}{4}. 
	\end{align*}
	Similarly, using Lemma \ref{lm:iid:column} with probability at least $1- \frac{\gamma^3}{4}$ for all $\i \in \{0,\dots,c-1\}$ we have
	\begin{align}\label{eq:Ralg:t1}
		\prob{ \exists_{i\in S} t_{\i+1} \leq x_i } \leq \frac{2 q}{\lambda_0} = \gamma.
	\end{align}
	Next, we prove the statement of the lemma assuming that for all $\i \in \{0,\dots,c-1\}$ Inequalities \ref{eq:Ralg:t1t2} and \ref{eq:Ralg:t1} hold. First note that we have
	\begin{align*}
		\prob{ t_{\i} \leq x_S < t_{\i+1}} &\leq 
		\prob{ \exists_{i\in S} t_{\i} \leq x_i < t_{\i+1}} \\
		&\leq (1+\gamma)q \sum_{i=1}^n p^{\i}_i. &\text{By Inequality \ref{eq:Ralg:t1t2}}
	\end{align*}
	This proves the upper bound. On the other hand we have
	\begin{align*}
		\prob{ t_{\i} \leq x_S < t_{\i+1}} & \geq 
		 \prob{ \nexists_{i\in S}  x_i \geq t_{\i+1} } \times \prob{ \exists_{i\in S} t_{\i} \leq x_i < t_{\i+1}} \\
		&= \Big(1-\prob{ \exists_{i\in S}  x_i \geq t_{\i+1} } \Big) \times \prob{ \exists_{i\in S} t_{\i} \leq x_i < t_{\i+1}}  \\
		&\geq (1-\gamma) \prob{ \exists_{i\in S} t_{\i} \leq x_i < t_{\i+1}} &\text{By Inequality \ref{eq:Ralg:t1}}\\
		&\geq (1-\gamma) (1-2\gamma)q \sum_{i=1}^n p^{\i}_i &\text{By Inequality \ref{eq:Ralg:t1t2}}\\
		& \geq (1-3\gamma)q \sum_{i=1}^n p^{\i}_i.
	\end{align*}
\end{proof}

%%%%%%%%%%%%%%%%%%%%%%%%%%%%%%%%%%%%%%%%%%%%%%%%%%%%%%%%%%%%%%%%%%%%%%%%%%

The following technical lemma will be useful for proving Lemma \ref{lm:iid:column} and Lemma \ref{lm:iid:row}.

\begin{lemma}\label{lm:iid:no2}
	Let $\chi_1,\dots,\chi_m$ be a sequence of independent binary random variables. We have
	\begin{align*}
	\prob{\sum_{i=1}^m \chi_i \geq 2} \leq \prob{\exists_{i} \chi_i=1}^2.
	\end{align*}
\end{lemma}
\begin{proof}
	We have
	\begin{align*}
	\prob{\sum_{i=1}^m \chi_i \geq 2} &= \sum_{j=1}^m \Big( \prob{\forall_{i<j} \chi_i=0} \prob{\chi_j = 1} \prob{\sum_{i=j+1}^m \chi_i \geq 1}\Big)\\
	&\leq \sum_{j=1}^m \Big( \prob{\forall_{i<j} \chi_i=0} \prob{\chi_j = 1} \prob{\sum_{i=0}^m \chi_i \geq 1}\Big)\\
	&= \prob{\sum_{i=0}^m \chi_i \geq 1} \sum_{j=1}^m \Big( \prob{\forall_{i<j} \chi_i=0} \prob{\chi_j = 1} \Big)\\
	& = \prob{\sum_{i=0}^m \chi_i \geq 1}^2\\
	&= \prob{\exists_{i} \chi_i=1}^2.
	\end{align*}
\end{proof}

%%%%%%%%%%%%%%%%%%%%%%%%%%%%%%%%%%%%%%%%%%%%%%%%%%%%%%%%%%%%%%%%%%%%%%%%%%

We can now prove Lemma~\ref{lm:iid:column}.
For a small set of indices $S$ chosen uniformly at random, we wish to upper bound the probability of observing at least two $x_i$'s with $i\in S$ above a threshold $\tau^0$. We declare this as a failure case in our algorithm in subsection \ref{subsec:reduc-alg}.  For convenience we restate as Lemma~\ref{lm:iid:column:proof} below.

\begin{lemma}\label{lm:iid:column:proof}
	Consider arbitrary numbers $\lambda_0,\gamma,\delta,q \in (0,1)$, $\rho \in (0,1-\lambda_0)$. 
	Set $\eps = \frac{  \gamma^2 q^2 \rho \lambda_0 }{2\log \frac {2} {\delta}}$.
	Let $S$ be a set of size $qn$, chosen uniformly at random without replacement from $D_1,\dots,D_n$. 
	Let $\tau^0$ be such that $\prob{\maxall \leq \tau^0}=1-\rho$. 
	Let $y_i$ be a random binary variable that is $1$ if $\tau^0\leq x_i$ and $0$ otherwise. Let $p'_i= \prob{y_i = 1}$. Assuming the no $\eps$-superstars assumption, with probability $1-\delta$ we have 
	\begin{align*}
		\prob{\exists_{i\in S}  y_i = 1}&\leq \frac{2 q}{\lambda_0}\\
			&\text{and}\\
		\prob{\sum_{i\in S} y_i \geq 2}  &\leq \frac{4 q^2}{\lambda_0^2}.
	\end{align*}
\end{lemma}
\begin{proof}
	With probability $1-2\exp \Big(-  \frac{  \gamma^2 q^2 \rho (1-(\lambda+\rho))}{2\eps }\Big) \geq 1-\delta$ (see Inequality \ref{eq:iid:delta}), Lemma \ref{lm:iid:sumBound} holds and hence we have
	\begin{align*}
	\prob{\exists_{i\in S} y_i = 1}
	&\leq \sum_{i\in S} p'_i  &\text{By Lemma \ref{lm:iid:ES}}\\
	&\leq 2 q \sum_{i=1}^n p'_i  &\text{By Lemma \ref{lm:iid:sumBound} with $\gamma < 1$}\\
	&\leq 2 q \sum_{i=1}^n \frac{\prob{i = \argmax_{j=1}^n x_j}}{\lambda_0}  &\text{By Lemma \ref{lm:iid:oneyprob} with $\lambda=0, \rho = 1-\lambda_0$}\\
	&\leq \frac{2 q}{\lambda_0}. &\sum_{i=1}^n \prob{i = \argmax_{j=1}^n x_j} = 1
	\end{align*}
	and hence, we have
	\begin{align*}
		\prob{\sum_{i\in S} y_i \geq 2} &\leq
		\prob{\exists_{i\in S} y_i = 1}^2 &\text{By Lemma \ref{lm:iid:no2}}\\
		%&\leq \frac{2 q}{\lambda_0}\prob{\exists_{i\in S} y_i = 1}\\
		&\leq \frac{4 q^2}{\lambda_0^2}.
	\end{align*}
\end{proof}

%%%%%%%%%%%%%%%%%%%%%%%%%%%%%%%%%%%%%%%%%%%%%%%%%%%%%%%%%%%%%%%%%%%%%%%%%%

Finally we will prove Lemma~\ref{lm:iid:row}.
We wish to upper bound the probability of observing at least two $x_i$'s within a narrow range $[\tau^0,\tau^1]$. We declare this as a failure case in our algorithm in subsection \ref{subsec:reduc-alg}.  For convenience we restate as Lemma~\ref{lm:iid:row:proof} below.

\begin{lemma}\label{lm:iid:row:proof}
	Consider arbitrary numbers $\rho,\lambda_0 \in (0,1)$ and $\lambda \in [0, 1-(\lambda_0+\rho)]$.
	Let $\tau^0$ and $\tau^1$ be such that $\prob{\maxall \leq \tau^0}=1-(\lambda+\rho)$ and $\prob{\maxall \leq \tau^1}=1-\lambda$. 
	Let $y_i$ be a random binary variable that is $1$ if $\tau^0\leq x_i \leq \tau^1$ and $0$ otherwise. 
	We have
	\begin{align*}
		\prob{\sum_{i=1}^n y_i \geq 2} \leq \frac{\rho^2}{\lambda_0^2}. 	
	\end{align*}
\end{lemma}
\begin{proof}
	We have
	\begin{align*}
		\prob{\sum_{i=1}^n y_i \geq 2} &\leq
		\prob{\exists_{i=1}^n y_i = 1}^2 &\text{By Lemma \ref{lm:iid:no2}}\\
		&\leq \Big(\frac{\rho}{1-\lambda}\Big) &\text{By Lemma \ref{lm:iid:allyprob}}\\	
		&\leq \frac{\rho^2}{\lambda_0^2}. 	
	\end{align*}
\end{proof}

%%%%%%%%%%%%%%%%%%%%%%%%%%%%%%%%%%%%%%%%%%%%%%%%%%%%%%%%%%%%%%%%%%%%%%%%%%
%%%%%%%%%%%%%%%%%%%%%%%%%%%%%%%%%%%%%%%%%%%%%%%%%%%%%%%%%%%%%%%%%%%%%%%%%%
%%%%%%%%%%%%%%%%%%%%%%%%%%%%%%%%%%%%%%%%%%%%%%%%%%%%%%%%%%%%%%%%%%%%%%%%%%

\end{document}